\documentclass[prx,twocolumn,superscriptaddress, longbibliography]{revtex4}
\usepackage{optidef}

\usepackage[SquareTraceBrackets]{quantum}
\usepackage{graphicx,bm,natbib,upgreek,amsmath,mathrsfs,accents}
\usepackage{amsbsy}
\usepackage[dvipsnames]{xcolor}
\definecolor{myblue}{named}{MidnightBlue}
\definecolor{mygreen}{RGB}{0,120,0}
\usepackage[colorlinks=true,citecolor=myblue,linkcolor=myblue,urlcolor=myblue]{hyperref}
\usepackage[T1]{fontenc}
\usepackage{newtxtext,newtxmath}

\usepackage[scaled]{helvet}
\usepackage{tikz}
\usetikzlibrary{arrows,decorations.pathmorphing,backgrounds,positioning,fit,petri}
\usepackage{gensymb}				
\usepackage[caption=false]{subfig}
\usepackage{tabularx}

\DeclareMathAlphabet{\mathscrbf}{OMS}{mdugm}{b}{n} 
\DeclareMathAlphabet\mathcalbf{OMS}{cmsy}{b}{n}		

\newcommand{\red}[1]{\textcolor{red}{#1}} 
\usepackage{amsthm}

\newtheorem{theorem}{Theorem} 
\newtheorem{lemma}[theorem]{Lemma}

\newtheorem{proposition}[theorem]{Proposition}

\def\wt{\rm wt} 
\def\even{\rm even} 
\def\odd{\rm odd} 
\def\>{\rangle} 
\def\<{\langle}

 \def\trnorm#1{{
 \left \| #1 \right \|_1   }} 

\DeclareGraphicsExtensions{.pdf, .jpg, .eps, .svg}

\DeclareMathOperator{\vect}{vec}

\makeatletter
\def\thickhline{%
  \noalign{\ifnum0=`}\fi\hrule \@height \thickarrayrulewidth \futurelet
   \reserved@a\@xthickhline}
\def\@xthickhline{\ifx\reserved@a\thickhline
               \vskip\doublerulesep
               \vskip-\thickarrayrulewidth
             \fi
      \ifnum0=`{\fi}}
\makeatother

\begin{document}

\title{Tight bounds on the simultaneous estimation of incompatible parameters}
\author{Jasminder S. Sidhu}
\email{jsmdrsidhu@gmail.com}
\affiliation{Department of Physics and Astronomy, The University of Sheffield, Sheffield, S3 7RH, UK}
\affiliation{SUPA Department of Physics, The University of Strathclyde, Glasgow, G4 0NG, UK}
\author{Yingkai Ouyang}
\email{y.ouyang@sheffield.ac.uk}
\affiliation{Department of Physics and Astronomy, The University of Sheffield, Sheffield, S3 7RH, UK}
\author{Earl T. Campbell}
\affiliation{Department of Physics and Astronomy, The University of Sheffield, Sheffield, S3 7RH, UK}
\affiliation{AWS Center for Quantum Computing, Pasadena, CA 91125, USA}
\author{Pieter Kok}%
\affiliation{Department of Physics and Astronomy, The University of Sheffield, Sheffield, S3 7RH, UK}
\date{\today}

\begin{abstract}
The estimation of multiple parameters in quantum metrology is important for a vast array of applications in quantum information processing. However, the unattainability of fundamental precision bounds for incompatible observables has greatly diminished the applicability of estimation theory in many practical implementations. The Holevo Cram{\'e}r-Rao bound (\textsc{hcrb}) provides the most fundamental, simultaneously attainable bound for multi-parameter estimation problems. A general closed form for the \textsc{hcrb} is not known given that it requires a complex optimisation over multiple variables. In this work, we develop an analytic approach to solving the \textsc{hcrb} for two parameters. Our analysis reveals the role of the \textsc{hcrb} and its interplay with alternative bounds in estimation theory. For more parameters, we generate a lower bound to the \textsc{hcrb}. Our work greatly reduces the complexity of determining the \textsc{hcrb} to solving a set of linear equations that even numerically permits a quadratic speedup over previous state-of-the-art approaches. We apply our results to compare the performance of different probe states in magnetic field sensing,
and characterise the performance of state tomography on the codespace of noisy bosonic error-correcting codes. The sensitivity of state tomography on noisy binomial codestates can be improved by tuning two coding parameters that relate to the number of correctable phase and amplitude damping errors. Our work provides fundamental insights and makes significant progress towards the estimation of multiple incompatible observables.
\end{abstract}

\maketitle

\section{Introduction}
\label{sec:introduction}


\noindent
Physical quantities such as time, phase, and entanglement cannot be measured directly, but instead must be inferred through indirect measurements. An important category of such indirect measurements is parameter estimation. Quantum metrology describes the quantum mechanical framework that handles this estimation procedure. By recasting the problem as a statistical inference problem, parameter estimation can be associated with fundamental precision bounds. The key question in quantum metrology is what is the fundamental precision bound and how we can achieve it. Early applications of estimation theory focused on single parameter estimation such as phase measurements~\cite{Braunstein94, Giovannetti2011_NP, Toth14}. The ultimate precision bound for a single parameter is the quantum Cram{\'e}r-Rao lower bound (\textsc{qcrb}), which was proved by Helstrom and Holevo~\cite{Helstrom1968_IEEE,Helstrom1976,Holevo2011}.  Multi-parameter quantum metrology extends the single parameter case~\cite{Matsumoto2002_JPA, asymptotic_theory_book_2005,Gill2012_arxiv}, and is of fundamental importance in understanding a variety of practical applications, such as Hamiltonian tomography~\cite{Burgarth2017_PRL}, field sensing~\cite{Tsang2011_PRL, Ang2013_NJP, Baumgratz2016_PRL} and imaging~\cite{Tsang2009_PRL, Humphreys2013_PRL,Gagatsos2016_PRA,Pezze2017_PRL,Rubio2020_PRA}, and distributed sensing~\cite{Proctor2017_arxiv, Sidhu2017_PRA, Sidhu2018_arxiv,Rubio2020_JPA}. A central problem is to determine the optimal measurement strategies that saturate the \textsc{qcrb}~\cite{Paris2009_IJQI}. To achieve this, one must assume locally unbiased estimators~\cite{Fraser1964_JRSS}, which is reasonable given large amounts of prior information~\cite{Hall2012_PRX,DemkowiczDobrzanski15}, and with many independent repetitions of the experiment~\cite{Kay1993}. Several reviews on the topic highlight recent progress in the field~\cite{Pezze2014_AI, Degen2017_RMP,Sidhu2020_AVS}.

Each individual parameter we wish to estimate has an optimal measurement observable. However, when we wish to estimate two or more parameters simultaneously, the corresponding optimal observables may be incompatible. In this case, we can not achieve the optimal precision for each parameter individually. In this case the \textsc{qcrb} matrix bound is generally not simultaneously saturable for all parameters~\cite{Zhu2015_SR, Heinosaari2016_JPA,Ragy2016_PRA}. This motivates the search for tighter bounds that can be realised for practical applications of multi-parameter estimation theory. The Holevo Cram{\'e}r Rao bound (\textsc{hcrb}) encapsulates the difficulties associated with incompatible observables~\cite{Holevo2011}. It represents the best precision attainable with collective measurements on an asymptotically large number of identical copies of a quantum state~\cite{Guta2006_PRA,Hayashi2008_JMP,Yamagata2013_AS,Yang2019_CMP}.

Despite its importance, the \textsc{hcrb} has seen limited use in quantum metrology so far. There are several reasons for this. First, the \textsc{hcrb} is difficult to evaluate given that it is defined through a complex optimisation over a set of observables. Second, implementing collective measurements is generally a difficult task. Nevertheless, applications of the \textsc{hcrb} in metrological tasks do exist. Suzuki found closed form results for parameter estimation with qubits~\cite{Suzuki2016_JMP}, and explored connections between different types of metrological bounds in the special case of two parameter estimation theory. For pure states~\cite{Matsumoto2002_JPA} and displacement estimation with Gaussian states~\cite{Holevo2011}, it has been shown that the \textsc{hcrb} is attained by single-copy measurements.
The \textsc{hcrb} was also used as a tool to define the precision of state estimation for finite dimensional quantum systems~\cite{Yang2019_CMP}. Bradshaw \emph{et al.} calculated the \textsc{hcrb} for a joint parameter estimation of a displacement operation on a pure two-mode squeezed probe~\cite{Bradshaw2017_PLA}.

Arguably, the \textsc{hcrb} is most relevant in multi-parameter estimation. An increasing number of true multi-parameter estimation protocols has been explored~\cite{Knott2016_PRA, Rehacek17,Zhang17, Proctor2018_PRL}, and therefore the need for general, attainable bounds on multi-parameter quantum estimation is urgent. Recently, Albarelli \emph{et al.} have investigated the numerical tractability of calculating the \textsc{hcrb} for the simultaneous estimation of multiple parameters~\cite{Albarelli2019_PRL}. For finite-dimensional systems, they recast the evaluation of the \textsc{hcrb} as a semi-definite program, which is an optimisation problem that can be efficiently implemented. To date, no general analytic expression for the \textsc{hcrb} is known.

In this paper, we find that it is possible to recast the \textsc{hcrb} as a quadratic program with linear constraints, thereby providing tight bounds for multi-parameter estimation problems. We develop an analytical approach to solving the two-parameter \textsc{hcrb}, and provide expressions on when the analytical solution is tight.
Our analytical solution for the optimal observables that can saturate the \textsc{hcrb}
allows one to establish analytically the minimum penalty due to the incompatibility of the observables. 
Specifically, we generalise attainability constraints for simultaneous multi-parameter estimation problems where the commonly used Cram{\'e}r-Rao bounds can not be saturated due to incompatibility. The analytic two-parameter \textsc{hcrb} can be considered a generalised quantum uncertainty relation~\cite{Frowis2015_PRA}. For more than two parameters, our method does not provide tight bounds but still outperforms the \textsc{qcrb}.


\subsection{Summary of results}
\label{subsec:results_summary}

\noindent
The \textsc{hcrb} is defined as a constrained minimisation problem over measurement observables. By recasting the definition as a quadratic program with linear constraints, we find exact solutions to this minimisation and determine the optimal observables. Our method to solve this minimisation relies on the notion of duality in optimisation theory, where the primal problem is transformed to its dual problem. Through the duality gap, we are able to quantify the minimum penalty of estimating incompatible observables.

In this article, we introduce three new algorithms that derive bounds to the \textsc{hcrb} for two and more parameters. First we determine upper and lower bounds to the \textsc{hcrb} for two-parameter estimation problems that are not always tight. This leads to simple analytic expressions that are straightforwardly determined for probe states with full rank. The salient feature of this framework, from which the simplification is inherited, is that only the boundary values for the Lagrange dual variables are considered. Second, this method is extended to determine upper and lower bounds to the \textsc{hcrb} for more than two parameters. Finally, we return to the two-parameter \textsc{hcrb} to develop tight bounds. For this, we lift the constraint on the values for the Lagrange dual variables to explore the full generality permitted by our method. Our analysis for this shows that the \textsc{hcrb} is a general solution to a Sylvester equation in the measurement observables, and recovers the standard Lyapunov \textsc{sld qcrb} solution when the weak commutativity criterion is violated. This algorithm can be implemented numerically using a Bartel-Stewart's algorithm for linear equation solvers, and offers a quadratic speedup in runtime over state-of-the-art semi-definite programming approaches.

Table~\ref{tab:hcrb_results} provides a high level summary of these algorithms, along with any assumptions made. Our results provide a significant extension of the capabilities of previous approaches, and clarifies the role of the \textsc{hcrb} in the estimation of incompatible observables.

\newlength{\thickarrayrulewidth}
\setlength{\thickarrayrulewidth}{2.1\arrayrulewidth}
\renewcommand{\arraystretch}{1.86}
\setlength{\tabcolsep}{8pt}
\begin{table*}[t!]
  \centering
  \begin{tabular}{>{\centering\arraybackslash}m{4cm}|>{\centering\arraybackslash}m{5cm}|>{\centering\arraybackslash}m{7cm}}
    \thickhline
    \textbf{Nature of HCRB bound}\vspace{2pt} & \textbf{Assumptions}\vspace{2pt} & \textbf{Algorithm details} \vspace{2pt}\\
    \hline
    Algorithm~\ref{algo:Algo_hcrb2_2}: analytic two-parameter bound. & Full rank $\rho$, linearly independent $\rho_1$ and $\rho_2$, and analytic form for $Q$-matrix. & Provides upper $\mathscr{U}$ and lower $\mathscr{L}$ bounds, need not be tight. \\
    Algorithms~\ref{algo:hcrb_gt} and~\ref{algo:hcrb_multi}: hybrid multi-parameter bound. & Full rank $\rho$ and analytic form for $Q$-matrix. & Provides upper $\mathscr{U}$ and lower $\mathscr{L}$ bounds, need not be tight. \\
    Algorithm~\ref{fig:masteralgo}: hybrid two-parameter bound.  &  Full rank $\rho$, spectral decomposition of $\rho$.   Full rankness of intermediate. $\mathcal{Q}$-matrix is full rank and takes at most $\tau$ time to compute.
    & Analytic bounds for $u \in [0,1]$. Tight bounds certifiably attained by numerically varying $u$ to maximise $\mathscr{L}_u$. Computes in $\mathcal{O}({\rm polylog}(1/\epsilon )\tau D^{0})$. \\
    Numerical two-parameter bound using Eq.~\eqref{eqn:Ysol_main}. & Full rank $\rho$.  &  Computes in $\mathcal{O}({\rm polylog}(1/\epsilon ) D^{2.376})$ time using Bartel-Stewart's algorithm,
    or $\mathcal{O}({\rm polylog}(1/\epsilon ) D^3)$ time using Gaussian elimination. 
    \\
    SDP numerical algorithm. & Arbitrary $\rho$, $\rho_1$, and $\rho_2$. &  Computes in $\mathcal{O}({\rm polylog}(1/\epsilon ) D^{2 \times 2.376})$ time,
    or $\mathcal{O}({\rm polylog}(1/\epsilon ) D^6)$ time using Gaussian elimination.\\
    \thickhline
  \end{tabular}
  \caption{Algorithms and bounds to the \textsc{hcrb} for two and more parameters. Bounds are analytic, numerical, or a hybrid of analytical and numerical, as indicated in the first column. There is a trade-off between the assumptions taken for each algorithm and its complexity. Here $D$ is the dimension of the probe state, and $\epsilon$ a measure of how close the \textsc{hcrb} bound is to optimal. The final row provides comparitive details for the \textsc{sdp} approach in Ref~\cite{Albarelli2019_PRL}.}
  \label{tab:hcrb_results}
\end{table*}%
%


\subsection{Outline of paper}
\label{subsec:paper_outline}

\noindent
We begin in section~\ref{sec:qet} by providing an overview of multi-parameter quantum estimation. In section~\ref{sec:hcrb}, we introduce the four new algorithms for analytic and numerical results to the \textsc{hcrb} for two parameters and arbitrary number of parameters. We detail connections between alternative precision bounds and significantly extend the capabilities of previous approaches in the literature. Section~\ref{sec:applications} discusses applications of our results to magnetometry and explores how bosonic quantum codes can bestow resilience of parameter estimates against noise beyond practical control. These applications demonstrate the strengths of our results and extend deep connections between quantum metrology and quantum error correcting codes. Finally, conclusions and interesting extensions to our results are provided in section~\ref{sec:conclusions}.


\section{Multi-parameter quantum estimation}
\label{sec:qet}

\noindent
Quantum estimation theory provides fundamental bounds to the estimation precision of physical parameters and the optimal measurements that saturate these limits~\cite{Paris2009_IJQI}. We are interested in estimating multiple parameters simultaneously. The prototypical scheme requires that the vector of parameters $\smash{\bm{\theta} = (\theta_1,\ldots,\theta_d)^\top \in \mathbb{R}^d}$ be imprinted on a quantum state $\rho(\bm{\theta})$.
Denoting $\mathbb H_D$ as the set of all Hermitian matrices in the Hilbert space of dimension $D$, we can see that $\rho(\bm{\theta})$ is a positive semidefinite matrix in $\mathbb H_D$ with unit trace. We define measurement operators via a positive operator valued measure (\textsc{povm})
\begin{align}
    \bm{\Pi} = \left\{\Pi_\omega \geq 0, \omega \in \Omega \, \vert \sum_{\omega \in \Omega} \Pi_\omega = \mathbbm{1}_D\right\},
    \label{eqn:povm}
\end{align}
where $\mathbbm{1}_D$ denotes the identity operator, and $\Omega$ is the set of measurement outcomes. The outcomes of such a measurement can be used in a function called the estimator $\smash{\check{\bm{\theta}}}$, which gives an estimate of the parameters. A general estimation scheme requires access to multiple identical copies of the quantum probe state. A separable measurement can be individually applied to each copy of the state to obtain estimates of each parameter separately, whereas a collective measurement can be applied jointly on all copies of the state to acquire a simultaneous estimate of all parameters. The ultimate precision bound is the one that is asymptotically achieved by a sequence of the best collective measurements as the number of copies tends to infinity~\cite{Gill2000_PRA,BarndorffNielsen2000, Bagan2004_PRA, Bagan2006_PRA,Hayashi2008_JMP,Guta2006_PRA,Kahn09,Yamagata2013_AS,Gill2012_arxiv}.

The performance of the estimator $\smash{\check{\bm{\theta}}}$ under any measurement can be quantified in terms of its mean square error (\textsc{mse}) matrix
\begin{align}
    \bm{\Sigma}_{\bm{\theta}}\left(\bm{\Pi}, \check{\bm{\theta}}\right) = \sum_{\bm{\omega} \in \Omega^N} p(\bm{\omega} \vert\bm{\theta}) \left(\check{\bm{\theta}}(\bm{\omega}) - \bm{\theta}\right) \left(\check{\bm{\theta}}(\bm{\omega}) - \bm{\theta}\right)^\top,
    \label{eqn:msem}
\end{align}
where the probability of measurement outcomes is provided by Born's rule $p(\omega \vert\bm{\theta}) = \tr{\rho(\bm{\theta})\Pi_\omega}$, and $N$ is the number of independently repeated measurements. The set of estimators are said to be locally unbiased if for all $\omega \in \Omega$
\begin{align}
\sum_{\bm{\omega} \in \Omega^N} \left(\check{\theta}_j(\bm{\omega}) - \theta_j\right) p(\bm{\omega} \vert\bm{\theta}) = 0, \; \sum_{\bm{\omega} \in \Omega^N}\check{\theta}_j(\bm{\omega})\partial_k p(\bm{\omega} \vert\bm{\theta}) = \delta_{jk},
\label{eqn:locally_unbiased_constraints}
\end{align}
where $\partial_k \equiv \partial/\partial \theta_k$. Under these conditions, the \textsc{mse} matrix is equivalent to the covariance matrix of parameter estimates, and is lower bounded through generalisations of the Cram{\'e}r-Rao bound from classical statistics
\begin{align}
    \bm{\Sigma}_{\bm{\theta}}\left(\bm{\Pi}, \check{\bm{\theta}}\right) \geq \mathcal{F}\left(\rho(\bm{\theta}), \bm{\Pi}\right)^{-1},
    \label{eqn:CRBM}
\end{align}
where $\mathcal{F}$ is the classical Fisher information matrix~\cite{Helstrom1976, Holevo2011}. The Fisher information characterises the \textsc{mse} matrix for the best classical data manipulation given a measurement strategy in the asymptotic limit~\cite{Kay1993}. A well known quantum generalisation includes the symmetric logarithmic derivative (\textsc{sld}), 
$L_j \in \mathbb H_N$, which is implicitly defined through $2\partial_j\rho=\{L_j,\rho\}$ and generates the real symmetric quantum Fisher information matrix (\textsc{qfim}) $\mathcal{I}_{jk}^\mathsf{S} = \text{Re}\left[\tr{\rho L_jL_k}\right]$~\cite{Helstrom67,Helstrom1968_IEEE}. This is referred to as the \textsc{sld qfim}. Notice that for ease of notation, we drop the explicit dependence of the state on the vector of parameters $\bm{\theta}$. Similarly, the right logarithmic derivative (\textsc{rld}), 
$R_j \in \mathbb H_N$, defined through $\partial_j\rho=\rho R_j$, induces the complex Hermitian \textsc{rld qfim} $\mathcal{I}_{jk}^\mathsf{R} = \tr{\rho R_jR_k}$~\cite{Yuen73,Belavkin1976_TMP}. 

There is a fundamental trade-off on how well each parameter in $\bm{\theta}$ can be simultaneously estimated. Hence, a meaningful multi-parameter estimation protocol minimises the weighted sum of parameter estimate variances. For this, a size $d$ positive definite square weight matrix $W$ is chosen to define the weighted mean square error (\textsc{wmse}) $\tr{W \bm{\Sigma}_{\bm{\theta}}}$. Holevo proved an equivalence between a matrix inequality and its corresponding scalar inequality, which allows the \textsc{wmse} to be optimally minimised. In particular, for any real symmetric $V$ and $W$, and Hermitian $M$, 
the inequality $V \geq M$ implies that 
$\tr{W V} \geq  \tr{\text{Re}[W M]} + {\rm Tr}\abs{\sqrt{W} \mathrm{Im}[M]\sqrt{W}}$~\cite[Lemma 6.6.1]{Holevo2011}, where $\text{Re}[\cdot]$ and $\text{Im}[\cdot]$ denote the real and imaginary part of each matrix element, and ${\rm Tr} \abs{\cdot}$ denotes the sum of the absolute values of the eigenvalues of a matrix
\footnote{One can also show this by first proving that $V\ge M$ implies that $\tr{ V} \geq  \tr{\text{Re}[ M]} + {\rm Tr}\abs{ \mathrm{Im}[M]}$, and subsequently replacing $V$ and $M$ by $\sqrt W V \sqrt W$ and $\sqrt W M \sqrt W$ respectively. The multiple on both sides by $\sqrt W$ is to ensure that $\sqrt W V \sqrt W$ is symmetric and $\sqrt W M \sqrt W$ is Hermitian.}
. 
Since the sum of the absolute values of the eigenvalues of a matrix is in fact the sum of the singular values of a matrix, the function ${\rm Tr} \abs{\cdot}$ is equivalent to the more commonly used trace norm $\|\cdot \|_1$.
Given that the covariance matrix is always real and symmetric, we identify the matrix $V$ with $\smash{\bm{\Sigma}_{\bm{\theta}}}$. Hence we can write the \textsc{wmse} as
\begin{align}
    \Tr{W\bm{\Sigma}_{\bm{\theta}}} \geq  \Tr{W \text{Re}[M]} + \bigl\|\sqrt{W} \mathrm{Im}[M] \sqrt{W}\bigr\|_1.
\label{eqn:scalar_Cost_func}
\end{align}
Notice that the scalar cost function in Eq.~\eqref{eqn:scalar_Cost_func} appropriately assigns individual priority weights to different parameters. For a given weight matrix and Hermitian matrix $M$, we want to minimise the scalar \textsc{wmse} to derive better parameter estimates.

Now, we identify $M$ with the inverse of the family of definitions for the quantum Fisher information matrices to generate different lower bounds on the scalar \textsc{wmse} cost function. Specifically, the matrices $\smash{\mathcal{I}^\mathsf{S}}$ and $\smash{\mathcal{I}^\mathsf{R}}$ generate the following scalar cost functions on the \textsc{sld qcrb} and \textsc{rld qcrb}
\begin{align}
    C_\mathsf{S}({\bm{\theta}}) &= \Tr{W[\mathcal{I}^\mathsf{S}]^{-1}}, \label{eqn:sld_sbound}\\
    C_\mathsf{R}({\bm{\theta}}) &= \Tr{W\text{Re}[\mathcal{I}^\mathsf{R}]^{-1}} + \bigl\|\sqrt{W}\text{Im}[\mathcal{I}^\mathsf{R}]^{-1} \sqrt{W}  \bigr\|_1,
    \label{eqn:scalar_bounds}
\end{align}
respectively. Nagaoka investigated in detail the relationship between these bounds~\cite{Nagaoka82}. The central problem in quantum estimation theory is the minimisation of these scalar bounds over the family of probability of distributions defined by quantum measurements. 


The \textsc{sld} and \textsc{rld qcrb} do not always provide the best bounds to parameter estimates. 
For example, the attainability of the \textsc{sld qcrb} does not generally hold for multiple parameter estimations~\cite{asymptotic_theory_book_2005}. Intuitively, any incompatibility among the parameters $\bm{\theta}$ prohibits the simultaneous optimal estimation of all parameters. 
Correspondingly, the \textsc{rld qcrb} is not always attainable since the optimal estimators derived from the \textsc{rld} may not correspond to physical \textsc{povm}s~\cite{Genoni2013_PRA}.  

The problem with saturability of the multiparameter bound was noted by Holevo, who provided the most general quantum extension to the classical Cram{\'e}r-Rao bound, called the Holevo Cram{\'e}r-Rao bound (\textsc{hcrb}). Specifically, if a vector of Hermitian observables $\smash{\bm{X} = (X_1, \ldots, X_d)}$ satisfies the locally unbiased conditions $\smash{\tr{\rho X_j} = 0}$ and $\smash{\tr{\partial_j \rho X_k} = \delta_{jk}}$, its covariance matrix $Z(\bm{X})$ with matrix elements $\smash{[Z(\bm{X})]_{jk} = \tr{ \rho X_j X_k }}$ satisfies the inequalities~\cite{Helstrom1976,Holevo2011}
\begin{align}
    Z(\bm{X}) \geq [\mathcal{I}^\mathsf{S}]^{-1}, \quad \text{and} \quad Z(\bm{X}) \geq [\mathcal{I}^\mathsf{R}]^{-1}.
\label{eqn:Obse_cov}
\end{align}
From this, it is clear that 
identifying $M$ in Eq.~\eqref{eqn:scalar_Cost_func} with the Hermitian matrix $Z(\bm{X})$, such that 
\begin{align}
 \mathrm{Tr}[W \bm{\Sigma}_{\bm{\theta}}]   &  \geq  \mathrm{Tr}[W {\rm Re}[Z(\bm{X})]] + \| \sqrt{W} \mathrm{Im} [Z(\bm{X})]\sqrt{W} \|_1,
 \label{preHolevo}
\end{align}
we have a tighter bound on the scalar \textsc{wmse} than either of the bounds in Eq.~\eqref{eqn:sld_sbound} or Eq.~\eqref{eqn:scalar_bounds}.
By optimising the objective function in Eq.~\eqref{preHolevo} subject to appropriate unbiasedness constraints on $X$, we obtain the tightest bound on the \textsc{wmse}. This optimisation defines the \textsc{hcrb}, $C_\mathsf{H}(\bm{\theta})$, which explicitly is the minimum
of the following minimisation problem~\cite{Nagaoka82}
\begin{mini} 
{X_1,\ldots,X_d}{\tr{ W{\rm Re}Z(\bm{X})} + \|\sqrt{W} {\rm Im}[ Z(\bm{X})] \sqrt{W} \|_1  ,}
{}{}
\addConstraint{\Tr{\rho X_j}}{=0}
\addConstraint{\Tr{\partial_j \rho X_k}}{=\delta_{jk}.}
\label{opt0}
\end{mini}
The \textsc{hcrb} is the best asymptotically attainable precision with global, unbiased measurements of a set of parameters. By minimising over only the first term in the objective function of Eq.~\eqref{opt0}, we obtain the \textsc{sld qcrb}~\cite{Nagaoka82}
\begin{align}
    C_\mathsf{S}({\bm{\theta}}) = \Tr{W[\mathcal{I}^\mathsf{S}]^{-1}} = \min_{X_1,\ldots, X_d} \tr{ W{\rm Re}Z}.
    \label{eqn:sld_qcrb_minimisation_form}
\end{align}
This shows that the \textsc{hcrb} is a tighter bound than the \textsc{sld qcrb}, since the second term in Eq.~\eqref{opt0} is non-negative~\cite{Ragy2016_PRA}. In fact, the \textsc{hcrb} is more informative than both the scalar \textsc{sld} and \textsc{rld} \textsc{qcrb}s, and satisfies the inequality~\cite{Holevo82}
\begin{align}
    \Tr{W \bm{\Sigma}_{\bm{\theta}}} \ge C_\mathsf{H}({\bm{\theta}}) \ge \text{max}\left\{C_\mathsf{S}({\bm{\theta}}), C_\mathsf{R}({\bm{\theta}})\right\}, 
    \label{eqn:Holevo_bound_intro}
\end{align}
and gives
the best asymptotically attainable precision with global, unbiased measurements of a set of parameters. Specifically, Helstrom~\cite{Helstrom1976} and Holevo~\cite{Holevo82} demonstrated that $C_\mathsf{H}({\bm{\theta}})$ is attainable if the locally unbiased equality constraints in Eq.~\eqref{eqn:locally_unbiased_constraints} are satisfied.

We note that the \textsc{hcrb} is not defined explicitly in terms of a closed form for a given statistical model. This is in contrast to the classical case, where the Fisher information can be readily determined from a given statistical model. Recent efforts have focused on determining upper bounds to the \textsc{hcrb}~\cite{Carollo2019_JSM,Tsang2019_arxiv, Albarelli2019_arxiv}. Specifically, the \textsc{hcrb} is upper bounded by a quantity that is twice the \textsc{sld-qcrb}, such that~\cite{Carollo2019_JSM}
\begin{align}
    \text{max}\left\{C_\mathsf{S}({\bm{\theta}}), C_\mathsf{R}({\bm{\theta}})\right\} \leq C_\mathsf{H}({\bm{\theta}}) \leq 2 C_\mathsf{S}({\bm{\theta}}).
    \label{eqn:holevo_sandwich}
\end{align}
In this paper, we provide an analytic solution to the \textsc{hcrb} and provide conditions on when it is tight.

The \textsc{hcrb} is the best asymptotically achievable bound under the conditions stated in Refs~\cite{Guta2006_PRA, Hayashi2008_JMP, Kahn09, Yamagata2013_AS}. Both inequalities in Eq.~\eqref{eqn:Holevo_bound_intro} can be tight~\cite{Ragy16}. For instance, consider the skew-symmetric matrix $\smash{\text{Im}(\tr{L_jL_k\rho(\bm{\theta})})}$. When 
\begin{align}
    \text{Im}(\tr{L_jL_k\rho(\bm{\theta})}) = 0\, , 
    \label{eqn:weak_commutativity}
\end{align}
we have $C_\mathsf{H}({\bm{\theta}}) = C_\mathsf{S}({\bm{\theta}})$~\cite{asymptotic_theory_book_2005}.
This condition is referred to as the weak commutativity criterion~\cite{Suzuki2019_E}, and when it is fulfilled the \textsc{qcrb} is a good proxy for the \textsc{hcrb}. In the next section, we show how we can use methods from optimisation theory to address the minimisation over several Hermitian operators in the case where the weak commutativity criterion is not fulfilled.



\section{Holevo Cram{\'e}r-Rao Bound}
\label{sec:hcrb}

\noindent
In this section, we present algorithms to calculate bounds on the Holevo Cram\'er-Rao bound $C_\mathsf{H}$. We first derive simple analytic upper and lower bounds for $C_\mathsf{H}$ for two parameters in section~\ref{sec:two_parameter_Setting}. We show how these bounds are generated by studying the optimisation problem using the method of Lagrange multipliers. This has the advantage of reducing the complexity involved in evaluating bounds on $C_\mathsf{H}$ to that of solving two sets of linear equations. In section~\ref{sec:multi_parameter_setting}, we focus on deriving lower bounds on $C_\mathsf{H}$ for more than two parameters. At the expense of additional analysis, our formalism can be extended to also provide tight analytic solutions to the \textsc{hcrb}. We demonstrate this in section~\ref{eqn:new_thrm}, where we provide a complete exposition of analytic bounds on the two-parameter \textsc{hcrb}, and provide conditions for when they bounds are tight.


\subsection{Simple bounds in the two-parameter setting}
\label{sec:two_parameter_Setting}

\noindent
We first consider the \textsc{hcrb} for two parameters $\smash{\bm{\theta} = (\theta_1, \theta_2)^\top}$. To obtain simple analytic bounds to the \textsc{hcrb}, we must define the weight matrix $W$ for the scalar bound. For simplicity, we use the identity weight matrix and determine upper and lower bounds to the two-parameter \textsc{hcrb} using optimisation theory~\cite{nocedal2006numerical}. We want to solve the minimisation in Eq.~\eqref{opt0}, which is convex but not quadratic. Hence, we first manipulate Eq.~\eqref{opt0} into a quadratic form in the variables $X_1$ and $X_2$. Then, such an optimisation problem can be studied analytically using the method of Lagrange multipliers. 

Choosing $Y = X_1 + i X_2$, Eq.~\eqref{opt0} can be written as an optimisation program (see appendix~\ref{app:2_param})
\begin{mini} 
{Y,t}{ t,}
{}{}
\addConstraint{\Tr{ Y \rho Y^\dagger } \le t,\; \Tr{ Y^\dagger \rho Y } }{ \le t }
\addConstraint{ \tr{ \rho Y}=0, \;  \tr{ \partial_1  \rho Y} = 1, \;  \tr{ \partial_2  \rho Y}}{ = i. }
\label{main-opt1}
\end{mini}
Note that by considering both the real and imaginary parts of the above equality constraints, the actual number of real-valued equality constraints is six, which is consistent with the number of equality constraints corresponding to the minimisation in Eq.~\eqref{opt0}. Here $Y$ is optimised over all complex matrices of dimension $D$, and is in general not a Hermitian matrix. By mapping $Y$ and $t$ into a real vector ${\bf x}$, we cast this optimisation program into the standard form of
\begin{align}
    \min_{\bf x} \{f({\bf x}): c_i({\bf x}) \le 0 , h_i({\bf x}) = 0\},\label{func_minimisation}
\end{align} 
where $f({\bf x})$ is a real linear objective function, while $h_i({\bf x})$ and $c_i({\bf x})$ are the corresponding equality and inequality constraint functions that must also be real. Eq.~\eqref{main-opt1} is a convex program, since its equality constraints are linear and its inequality constraints are quadratic and convex. To check whether we can use optimality conditions from optimisation theory, we check whether Slater's constraint qualification holds. This amounts to checking that all the inequality constraints in Eq.~\eqref{main-opt1} can strictly hold. Since $t$ can be arbitrarily large, this indeed is the case. The optimality conditions for a continuous optimisation program are best stated in terms of the Lagrangian of Eq.~\eqref{main-opt1}, given by 
\begin{align}
\mathsf{L}({\bf x}, {\bm \lambda},{\bm z}) = f({\bf x}) + \sum_{i=1}^2 \lambda_i c_i({\bf x}) + \sum_{i=1}^6 z_i h_i({\bf x}),
\label{eqn:lagrangian}
\end{align}
where the coefficients $\lambda_i \geq 0$ and $z_i \in \mathbbm{R}$ are Lagrange multipliers for the inequality and equality constraints respectively. Since Eq.~\eqref{func_minimisation} is a convex program and Slater's constraint qualification holds, the first order Karush-Kuhn-Tucker (\textsc{kkt}) conditions of stationarity, primal and dual feasibility, and complementary slackness are necessary and sufficient~\cite{nocedal2006numerical} to determine the optimality of Eq.~\eqref{main-opt1}.

For our problem we have dual variables ${\bm \lambda} = (\lambda_1, \lambda_2) = (u,v)^\top$ and ${\bf z} = (z_1,\dots,z_6)^\top$, which are vectors of Lagrange multipliers. The primal variables are $Y$ and $t$, and the Lagrangian is given by
\begin{align}
\mathsf{L}(Y,t,u,v,{\bf z}) = &
\, \, t(1-u-v) - {\bf b}^\top {\bf z} + u \tr{Y \rho Y^\dagger}
\notag\\
 & + v \tr{Y^\dagger \rho Y}+\tr{A Y} + \tr{A^\dagger Y^\dagger}\, .
 \label{eqn:Lagrangian_bad_boy}
 \end{align}
Here $\smash{{\bf b} = (0,1,0,0,0,1)^\top}$ is a column vector that encodes the equality constraints in Eq.~\eqref{main-opt1}, constructed in Appendix~\ref{subsec:2param-lagrangian}. The operator $A$ is a linear superposition of $\rho$ and its derivatives, 
\begin{align}
 A = z_1 A_1 + \dots + z_6 A_6,
\label{eqn:A_defos}
\end{align}
where 
\begin{align}
\begin{split}
     A_1 = \frac12 \rho\, ,  & \qquad A_4 = -i A_1 , \cr
     A_2 = \frac12 \partial_1 \rho\, , & \qquad A_5 = -i A_2 , \cr
     A_3 = \frac12 \partial_2 \rho\, , & \qquad A_6 = -i A_3 .
\end{split}     
\label{def:Ajs}
\end{align}
Due to the duality principle in optimisation theory~\cite{nocedal2006numerical}, we may equivalently view the optimisation by considering the Lagrange dual function $g( {\bm \lambda},{\bf z}) = \inf_{\bf x} \mathsf{L}({\bf x},  {\bm \lambda},{\bf z})$ of Eq.~\eqref{eqn:lagrangian}. 
Since the Lagrangian $\mathsf{L}$ is quadratic in ${\bf x}$, 
the Lagrange dual can be found analytically 
by an unconstrained minimisation of the Lagrangian with respect to ${\bf x}$ 
for fixed values of the dual variables ${\bm \lambda}$ and ${\bf z}$~\cite{nocedal2006numerical}. 
Due to the structure of the Lagrangian in Eq.~\eqref{eqn:Lagrangian_bad_boy}, the Lagrange dual is never unbounded from below whenever $u+v=1$.  Hence, maximising the Lagrange dual function corresponds to an unconstrained maximisation problem. Since the Lagrange dual is also a quadratic function in terms of its dual variables ${\bf z}$, it can be easily maximised exactly with respect to ${\bf z}$.

Note that our Lagrange dual is not a quadratic function with respect to ${\bm \lambda} = (u,v)$. To bound $C_\mathsf{H}$, it suffices to evaluate the Lagrangian for feasible values of $(u,v)$ that satisfy $u+v=1$. Two such values are the boundary values $(u,v) = \{(0,1), (1,0)\}$, for which the Lagrangian in Eq.~\eqref{eqn:Lagrangian_bad_boy} is greatly simplified.
For each case, we first determine the stationary point of the resulting Lagrangian with respect to $Y$, where $Y$ has an implicit dependence on ${\bf z}$, and then perform a maximisation over ${\bf z}$. By evaluating the primal and dual objective functions, 
we obtain simple analytic two-sided bounds for $C_\mathsf{H}$.
Specifically, an analytic lower bound $\mathscr L$ to the \textsc{hcrb} is determined through finding $\smash{{\bf z} \in \mathbb R^6}$ that solves
\begin{align}
   2 {\rm Re}(Q_j) {\bf z} + {\bf b} = 0,\qquad j=1,2
\label{eqn:linear_eqn_set}
\end{align}
where $\smash{Q_j}$ has the matrix elements 
\begin{align}
 [Q_1]_{ik} = \tr{ A_i ^\dagger \rho^{-1}  A_k} \quad\text{and}\quad [Q_2]_{ik} = \tr{ A_i \rho^{-1}  A_k ^\dagger}.\label{def:Qjs}
\end{align}
Details for this are delegated to appendix~\ref{app:2_param}. The matrices ${\rm Re}(Q_j)$ are full rank when the derivatives $\partial_1\rho$ and $\partial_2\rho$ are linearly independent. Armed with these dual variables ${\bf z}$, we collect the result of this optimisation in the following theorem:%
\begin{theorem}%
\label{thm:2-param}%
Let $\partial_1\rho$ and $\partial_2 \rho$ be linearly independent. With $Q_j$ defined in Eq.~\eqref{def:Qjs} and the matrices $A_1,\dots,A_6$ given in Eq.~\eqref{def:Ajs}, the \textsc{hcrb} $C_\mathsf{H}$ for two parameters satisfies the inequality  
\begin{align*}
    \max_{j=1,2} \{l_j\} =\mathscr L
    \le C_\mathsf{H} \le  \mathscr U =
    \min_{j=1,2}\{ \max \{l_j,m_j \} \},
\end{align*}
where
\begin{align}
    l_j &=\frac{1}{4}{\bf b}^\top {\rm Re}(Q_j)^{-1} {\bf b},\\
    m_j &= \sum_{a,b=1}^6 \tr{ \rho^{-2} A_a  \rho A_b^\dagger } z_{a,j}z_{b,j}.
\end{align}
and
\begin{align}
    z_{a,j} = 
    -\frac{1}{2}\left([{\rm Re}(Q_j)^{-1}]_{a2}  +
    [{\rm Re}(Q_j)^{-1}]_{a6}\right).%
\end{align}%
\end{theorem}%
For a detailed proof of this theorem, consult appendices~\ref{subsec:2param-lagrangian}-\ref{subsec:full_rank_of_Q}.
Theorem~\ref{thm:2-param} gives a simple procedure for finding analytic upper and lower bounds to two-parameter \textsc{hcrb}. Notice that the complexity of determining these bounds are commensurate with linear equation solvers that are used in determining the Lagrange dual variables. This makes these bounds readily accessible for general two-parameter applications. Fig.~\ref{algo:Algo_hcrb2_2} shows the pseudocode for this procedure. 

%
\begin{figure}[t!]
\centering
\includegraphics[width =\columnwidth]{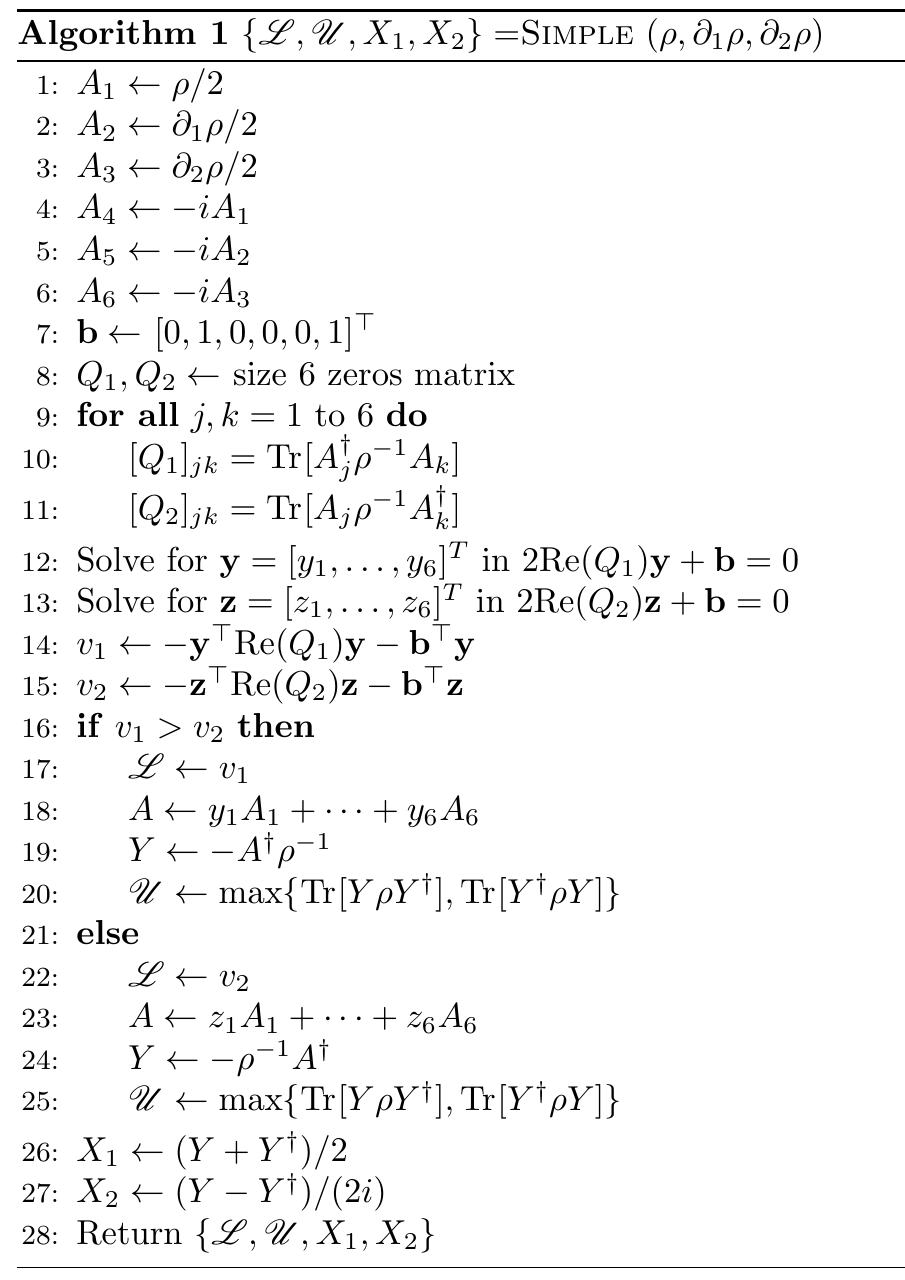}
\caption{Pseudocode to determine simple bounds to the two-parameter \textsc{hcrb} and its associated optimal measurement observables $X_1$ and $X_2$. Note that this algorithm depends only on the state $\rho$ and its two derivatives $\partial_1\rho$ and $\partial_2\rho$.}
\label{algo:Algo_hcrb2_2}
\end{figure}%

Next, we establish how to construct the observables that saturate these bounds. Specifically, there are two choices for $Y$ that minimise the Lagrangian in Eq.~\eqref{eqn:Lagrangian_bad_boy}, corresponding to the two choices for $(u, v)$:
\begin{align}
 (u,v) = (1,0) &: \quad Y = - A^\dagger \rho^{-1},\\
 (u,v) = (0,1) &: \quad Y = -\rho^{-1} A^\dagger,
\end{align}
that correspond to the choice $Q_1$ and $Q_2$, respectively. Then, using Eq.~\eqref{eqn:A_defos} and the optimised values for ${\bf z}$, we construct the analytic form for the observables $X_1$ and $X_2$. Since the matrices $Q_j, j \in \{1,2\}$ are only 6 dimensional matrices, determining $\smash{{\rm Re}(Q_j)}$ is easy and hence it is straightforward to find analytic bounds to the two-parameter \textsc{hcrb}. The procedure is shown algorithmically in Fig.~\ref{algo:Algo_hcrb2_2}.

Finally, we note that if our lower bound to $C_{\mathsf{H}}$ is strictly larger than $C_{\mathsf{S}}$, then 
we know that the skew-symmetric matrix ${\rm Im}(\tr{L_jL_k\rho(\bm \theta)})$ cannot be equal to zero, and the weak commutativity criterion does not hold. 


\subsection{Lower bound in the multi-parameter setting}
\label{sec:multi_parameter_setting}

\noindent
For more than two parameters, we can also use the method of Lagrange multipliers to bound the \textsc{hcrb}. However, this method is considerably more involved than the two-parameter case. In the two-parameter case, we could obtain a simple quadratic expression for ${\rm Re}\tr{Z} + \|{\rm Im Z}\|_1$ that appears in the objective function of Eq.~\eqref{opt0}. However for the corresponding generalisation to more parameters, ${\rm Re}\tr{Z} + \|{\rm Im Z}\|_1$ is no longer a quadratic form in the variables $X_j$. For example, for three parameters $Z$ takes the form
\begin{align}
    Z(\bm{X})  = 
    \begin{pmatrix}
    \tr{ \rho X_1^2 } & \tr{ \rho X_1 X_2 } &  \tr{ \rho X_1 X_3 }\\
    \tr{ \rho X_2 X_1 } & \tr{ \rho X_2^2 }  & \tr{ \rho X_2 X_3 }\\
    \tr{ \rho X_3 X_1 } & \tr{ \rho X_3 X_2 }  & \tr{ \rho X_3^2 }
    \end{pmatrix}.
\end{align}
The trace norm of ${\rm Im Z}$ is related to the eigenvalues of ${\rm Im Z}$, and the eigenvalues of a $3 \times 3$ matrix involve a cubic equation. This renders evaluating the trace norm incompatible with our methodology. To address this, we obtain a lower bound to $\|{\rm Im Z}\|_1$ that allows ${\rm Re}\tr{Z} + \|{\rm Im Z}\|_1$ to be written as a quadratic form. As shown in appendix~\ref{app:lower_bounds}, this yields an optimisation problem whose optimal value is a lower bound to the \textsc{hcrb}, and which is given by 
 \begin{align}
     \min\{t : \Tr{\rho X_j} = 0, \Tr{\partial_j \rho X_k} = \delta_{jk},  V_{\bm{\alpha}}  \le t \} ,
     \label{eq:main-text-opt1-d}
 \end{align}
where the minimisation is performed over $t$ and the Hermitian matrices $X_1,\dots, X_d$, with $j,k = \{1,\dots, d\}$, and the inequality constraint $V_{\bm{\alpha}}$ is a function of a binary string $\bm{\alpha}$ such that
\begin{align}
V_{\bm{\alpha}}
    &=
\frac 1 2 \sum_{j=1}^d
\tr{( X_j + (-1)^{\alpha_j} i X_{j+1}) 
\rho ( X_j + (-1)^{\alpha_j} i X_{j+1})^\dagger}\, ,
\label{main-text-eq:Vt-defi}
\end{align}
with $X_{d+1}=X_1$. The inequality constraints $V_{\bm{\alpha}}$ arise from the structure of our lower bound on the trace norm of ${\rm Im Z}$ (see appendix~\ref{app:lower_bounds}). By substituting $Y_j = \sum_{k=1}^d S_{jk} X_j$ where 
\begin{align}
 S= \begin{cases}
    \sum_{j \in \mathbb Z_d} 
    \left( |j\>\<j|  + i |j\>\<j \oplus 1|   \right)  & d \neq 0 \pmod{4} \\
     \sum_{j \in \mathbb Z_d} 
    \left( |j\>\<j|  + (-1)^{\delta_{j,d}}i|j\>\<j \oplus 1|   \right) & \mbox{otherwise }  \label{def:S-matrix_main},
    \end{cases}
\end{align}
we can write the matrices $X_j$ in terms of the matrices $Y_j$, as before. We next interpret the $Y_j$ as arbitrary complex matrices of size $n$, and impose Hermicity conditions for the corresponding $X_j$ matrices.

The Lagrangian of such an optimisation problem is a function of the complex matrices $\{Y_1,\dots, Y_d\}$, and also a function of its Lagrange multipliers. Its Lagrange multipliers are given by the non-negative multipliers ${\bf v}\in \mathbb R^{2^d}$ for the inequality constraints, ${\bf z} \in \mathbb R^{d(d+1)}$ for the equality constraints, and Hermitian multipliers $\xi_1,\dots, \xi_d$ for the Hermitian constraints. Most importantly, the inequality constraints can be satisfied strictly, so Slater's constraint qualification holds, and we can use the \textsc{kkt} to determine the optimality conditions for Eq.~\eqref{main-text-eq:Vt-defi}. We minimise the Lagrangian constructed from the optimisation problem in Eq.~\eqref{eq:main-text-opt1-d}. Since the Lagrangian is a convex quadratic form in the variables $Y_1,\dots,Y_d$, it can be minimised exactly. When this is done, we obtain the Lagrange dual function, which only depends on the Lagrange multipliers ${\bf v}, {\bf z}$ and $\xi_1,\dots, \xi_d$. 
The Lagrange dual function always gives a lower bound for the primal optimisation problem.

While the Lagrange dual is quadratic in ${\bf z}$ and $\xi_1,\dots, \xi_d$, it is not quadratic in ${\bf v}$. By minimising the Lagrangian over $t$ and using the \textsc{kkt} conditions, we conclude as before that the sum of the components in ${\bf v}$ is 1. We obtain a lower bound for the Lagrange dual by maximising over a discrete set of feasible Lagrange multipliers ${\bf v}$, which corresponds to the tightness of the constraints $V_{\bm{\alpha}} \le t$. Thus, we created a quadratic optimisation problem for three or more parameters that leads to a lower bound on $C_\mathsf{H}$. However, there is no guarantee that this lower bound is tight.

\begin{figure}[t!]
\centering
\includegraphics[width =\columnwidth]{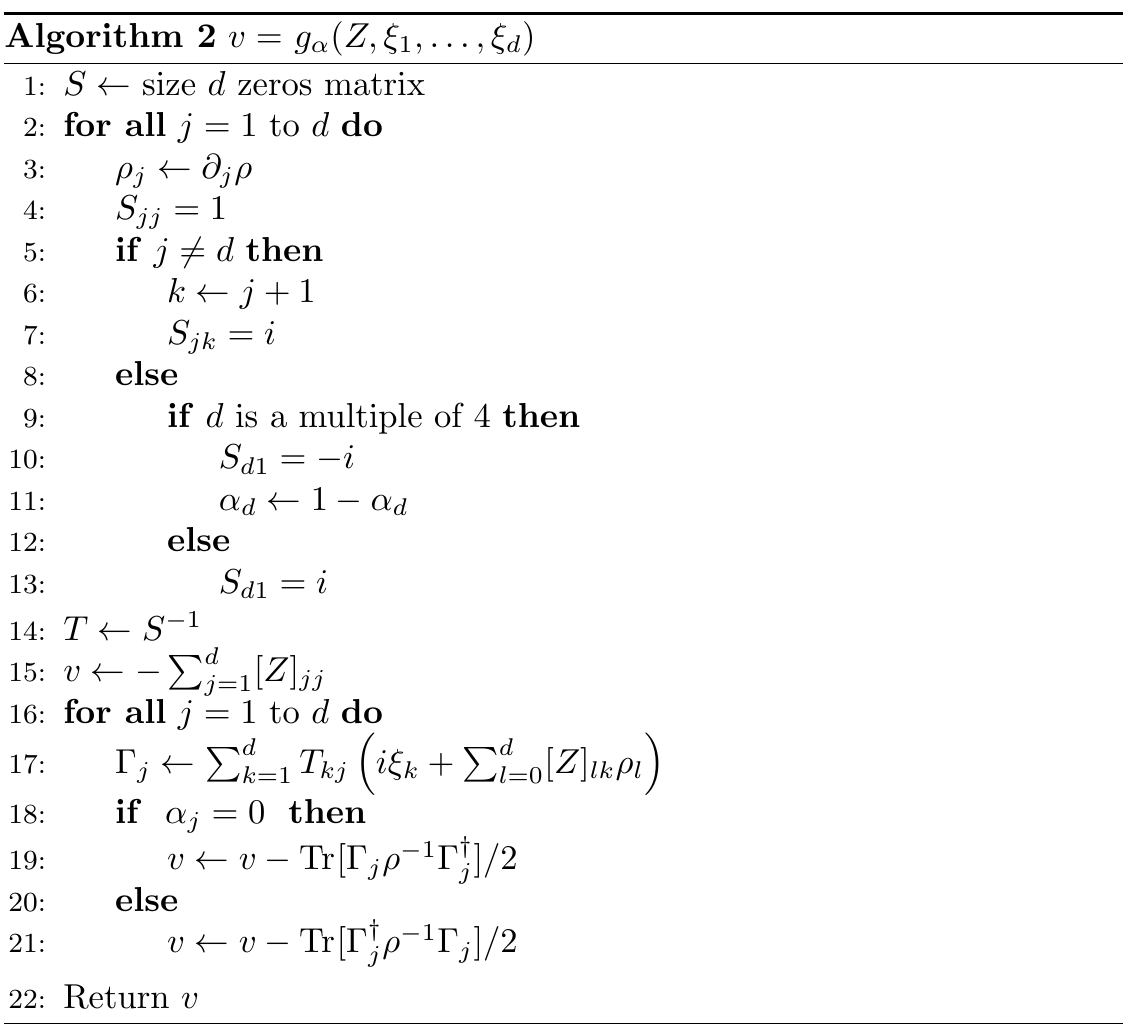}
\caption{Pseudocode to generate the Lagrange dual functions defined in Eq.~\eqref{eqn:gv_defo}.}
\label{algo:hcrb_gt}
\end{figure}

Next, we study the Lagrange dual function. By carefully choosing ${\bf v}$, the Lagrangian is quadratic in $\{Y_1, \ldots, Y_d\}$ and can be minimised individually for each $Y_j$. The coefficients for $Y_j$ in the Lagrangian are given by $\Gamma_j$, where 
\begin{align}
\Gamma_j &=
     \sum_{k=1}^d 
    T_{k,j}
    \left( 
        \sum_{l=0}^d Z_{l,k} \rho_l  
        +
        i \xi_k
    \right) ,
\end{align} 
where $\rho_0 = \rho$, $\rho_j = \partial_j\rho$ for $j=\{1,\ldots,d\}$, and $T_{k,j}$ are matrix elements that relate the  $Y_j$ to the $X_k$. Specifically, $T$ is the matrix inverse of $S$.
Then the optimal value for the Lagrange multipliers can be obtained by maximising the Lagrange dual functions 
\begin{align}
g_{\bm{\alpha}}
=&
- \sum_{j=1}^d z_{j,j} 
- \sum_{j=1}^d
\frac{
\delta_{0,\bar{\alpha}_j} \tr{\Gamma_j \rho^{-1} \Gamma_j^\dagger}
+
\delta_{1,\bar{\alpha}_j} \tr{\Gamma_j^\dagger \rho^{-1} \Gamma_j}
}{2}.
\label{eqn:gv_defo}
\end{align}
with respect to the scalar variables $z_{j,k}$ and the Hermitian variables $\xi_j$, where $\bar{\bm{\alpha}}=\bm{\alpha}$ when $d$ is not a multiple of 4, and when $d$ is a multiple of 4 then $\bar{\bm{\alpha}}$ differs from $\bm{\alpha}$ by simply flipping the last bit. Our lower bound to $C_\mathsf{H}$ in the multi-parameter setting is then given by the following theorem.
\begin{figure}[t!]
\centering
\includegraphics[width =\columnwidth]{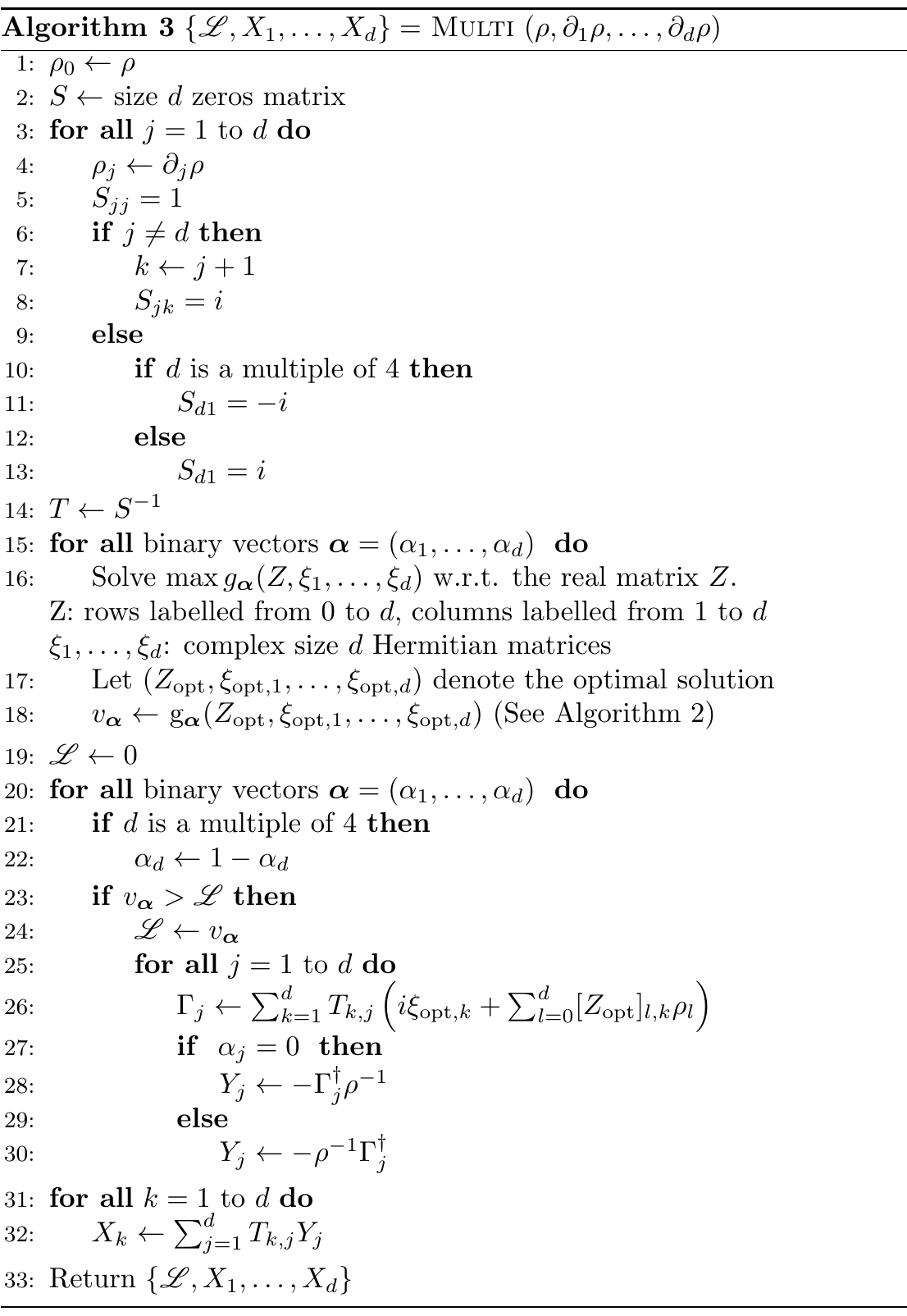}
\caption{Pseudocode to generate a lower bound to the \textsc{hcrb} for multiple parameters.}
\label{algo:hcrb_multi}
\end{figure} 
\begin{theorem}\label{thm:multi-param-lower-bound}
Let $d \ge 3$, ${\bf z} \in \mathbb R^{d(d+1)}$ and $\xi_1,\dots, \xi_d$ are Hermitian matrices. Then 
\begin{align*}
 C_\mathsf{H} \ge   \max_{{\bm{\alpha}} \in\{0,1\}^d }\max_{{\bf z}, \xi_1,\dots, \xi_d }  g_{\bm{\alpha}} ,
\end{align*}   
where $g_{\bm{\alpha}}$ is given by Eq.~\eqref{eqn:gv_defo}.
\end{theorem}
\noindent
This optimisation problem can be solved exactly using a single step of Newton's method. It requires the input state $\rho$ and its derivatives $\partial_j\rho$. The algorithm to implement this lower bound is illustrated in Fig.~\ref{algo:hcrb_multi}.



\subsection{Tight two-parameter bounds}
\label{eqn:new_thrm}

\noindent
Notice that for Theorem~\ref{thm:2-param}, we constrained the values of the Lagrange multiplier $u$ to two values. This does not provide the most general case and as a result, the analysis can generate observables that are not always optimal. That is, the corresponding upper and lower bounds are not always tight. By lifting this restriction, we expand the analysis to explore the full generality of our formalism to generate tight bounds to the estimation of incompatible observables. As we observe in this section, this is necessary to develop an intuition into multiparameter quantum estimation that is captured by the construction of the \textsc{hcrb}. To achieve this, we revisit the two-parameter scenario. Specifically, for fixed $u$, we minimise the Lagrangian and find the optimal observables that attain these stationary points.  In doing so, for every feasible value of $u$, we obtain an upper and lower bound on the \textsc{hcrb}. Since the lower bound is a concave and smooth function, then optimisation theory guarantees a solution to both the \textsc{hcrb} and the observables that attain it.

\begin{figure}[t!] \includegraphics[width=\linewidth]{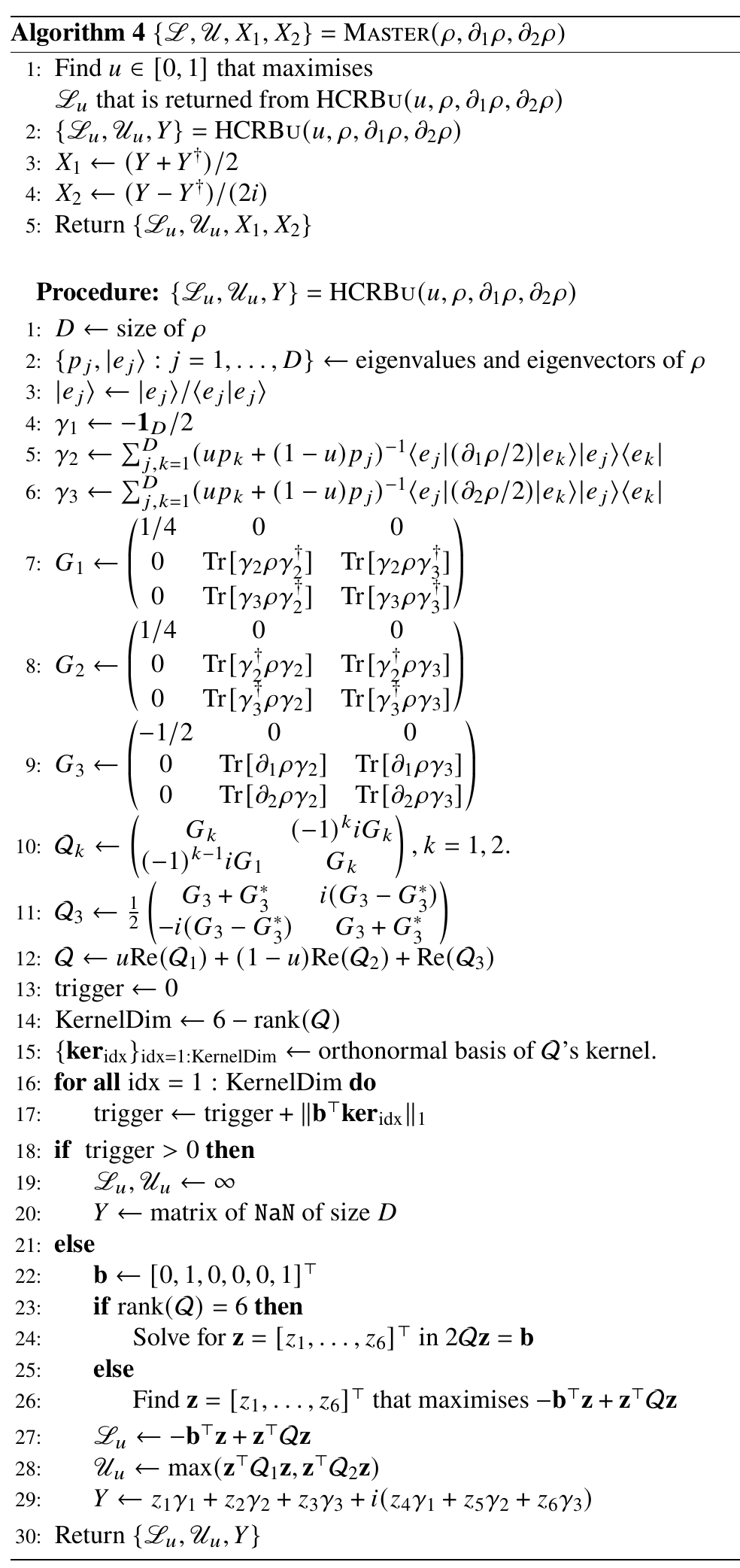}
\caption{Master algorithm to generate the analytic form of the two-parameter \textsc{hcrb}.} 
\label{fig:masteralgo}
\end{figure}%

Recall the Lagrangian in Eq.~\eqref{eqn:Lagrangian_bad_boy}, which for $u + v = 1$ becomes
\begin{align}
\begin{split}
    \mathsf{L}(Y,u,{\bf z}) =& -{\bf b}^\top {\bf z} + u \tr{Y \rho Y^\dagger} + (1-u) \tr{Y^\dagger \rho Y} \\
    &+ \tr{A Y} + \tr{A^\dagger Y^\dagger}.
\end{split}
\label{eqn:orig_Lagrangian_main}    
\end{align}
As we show in appendix~\ref{app:app_lag_min_arb_u}, this Lagrangian is minimised when $Y$ is chosen such that
\begin{align}
    u Y \rho + (1-u) \rho Y = - A^\dagger.
\label{eqn:general_lds}    
\end{align}    
Notice that $Y$, and hence the optimal observables, is the solution to a Sylvester equation. 
When ${\rm Im} Z = 0$, we know that the observables $X_1$ and $X_2$ commute and the weak commutativity criterion is preserved. Then, Eq.~\eqref{eqn:general_lds} reduces to solving a Lyapunov equation, where $Y$ generates the well explored \textsc{sld}. This corresponds to 
\begin{align}
    |\tr{\rho X_1 X_2} - \tr{\rho X_2 X_1}| = 0\label{eq:zero-gap}\, ,
\end{align} 
for the optimal $X_1$ and $X_2$, which recovers the weak commutativity condition. In this way, the solution to $Y$ in Eq.~\eqref{eqn:general_lds} provides the most general definition to quantum logarithmic derivatives for multiple incompatible parameters $\bm{\theta}$. From our definition of $A$, it defines exactly how the optimal observables depend on dynamics in both parameters. 

Similar to analytic solutions to the \textsc{sld}, given the spectral decomposition of the state $\rho = \sum_j p_j \ket{\smash{e_j}}\bra{\smash{e_j}}$, we can analytically solve Eq.~\eqref{eqn:general_lds} to obtain the $Y$ that minimises the Lagrangian:
\begin{align}
 Y = -\sum_{j,k} (u p_k + (1-u) p_j)^{-1} \bra{\smash{e_j}} A^{\dagger}\ket{e_k} \ket{\smash{e_j}}\bra{e_k}.   
\label{eqn:Ysol_main}  
\end{align}
By taking the trace, it is clear to see that $Y$ has a zero expectation value. This recovers the unbiasedness condition on the observables as required. Using $Y = X_1 + iX_2$ and the definition for $Y$ in Eq.~\eqref{eqn:Ysol_main}, we can write an analytic solution for the observables that saturate the \textsc{hcrb} in terms of the optimal Lagrange multipliers ${\bf z}$. Specifically, defining the statistical admixture $\smash{\varrho^{(s)} = \sum_{l=0}^2 z_{l+s}\rho_l}$, then in the eigenbasis of $\rho$, we have
\begin{align}
   [X_1]_{jk} &= \frac{i(p_j - p_k) (1-2u) [\varrho^{(4)}]_{jk} - (p_j + p_k) [\varrho^{(1)}]_{jk}}{4(u p_k + (1-u)p_j)(u p_j + (1-u) p_k)},\label{eqn:sols_for_x1}\\
    [X_2]_{jk} &= -\frac{i(p_j - p_k) (1-2u) [\varrho^{(1)}]_{jk} + (p_j + p_k) [\varrho^{(4)}]_{jk}}{4(u p_k + (1-u)p_j)(u p_j + (1-u) p_k)},
\label{eqn:sols_for_x}    
\end{align}
with $\smash{[\varrho^{(s)}]_{jk} = \braket{e_j\vert\varrho^{(s)}\vert e_k}}$. The Hermiticity of the state and its derivatives guarantees the Hermiticity of these observables such that $\smash{[X_l]_{jk} = [X_l]_{kj}^*}$. Equations~\eqref{eqn:sols_for_x1} and~\eqref{eqn:sols_for_x} shows exactly how each observable depends on the dynamics of each parameter.


With access to the spectral decomposition of $\rho$, we can also find analytic expressions to the \textsc{hcrb}. The master algorithm in Fig.~\ref{fig:masteralgo} concisely clarifies this procedure. With this procedure, theorem~\ref{theorem:thight_bounds} concretely demonstrates how to construct tight bounds on the \textsc{hcrb}, and is central to our result.


\begin{theorem}[]
\label{theorem:thight_bounds}
Let ${\mathcal Q_1}$, ${\mathcal Q_2}$ and ${\mathcal Q_3}$ be matrices defined in the Master algorithm in Fig.~\ref{fig:masteralgo}, and let ${\bf b} = (0,1,0,0,0,1)^\top$ be a column vector.
Let $0<u<1$ and ${\mathcal Q} = u{\rm Re}{\mathcal Q_1} + (1-u) {\rm Re}{\mathcal Q_2} + {\rm Re}{\mathcal Q_3}$ be a negative definite matrix. Then when ${\mathcal Q}$ is full rank for all $0<u<1$, the \textsc{hcrb} is bounded through $\smash{ \mathscr{L}_u \le C_{\mathsf{H}} \le  \mathscr{U}_u}$
where
\begin{align}
    \mathscr{L}_u &=  -\frac{1}{4} {\bf b}^\top  {\mathcal Q} ^{-1} {\bf b}\\
     \mathscr{U}_u &= \frac{1}{4} \max\left\{ 
      {\bf b}^{\top}{\mathcal Q} ^{-1} {\mathcal Q_1} {\mathcal Q} ^{-1} {\bf b} 
    , {\bf b}^{\top} {\mathcal Q} ^{-1} {\mathcal Q_2} {\mathcal Q} ^{-1} {\bf b} \right\}
\end{align}
with equality on both sides attained at the stationary point of the lower bound with respect to $u$ when  
$    \smash{\frac{d}{du}({\bf b}^\top  {\mathcal Q} ^{-1} {\bf b}) 
= 
  ({\bf b}^\top
  {\mathcal Q}^{-1} \frac{d{\mathcal Q}}{du} {\mathcal Q}^{-1} {\bf b}) 
= 0}$.
\label{thrm:hcrb_bounds_arb_u}
\end{theorem}
We refer the reader to Appendix~\ref{app:app_lag_min_arb_u} for a complete proof of this theorem. For applications where the spectral decomposition of the state is not known, the Sylvester equation~\eqref{eqn:Ysol_main} can be efficiently solved numerically using a variant of the Bartel-Stewart algorithm~\cite{Bartels1972_ACM}.

Before concluding, we clarify an important subtlety. Our analysis assumes that the probe state is fixed. However, there are multi-parameter sensing applications that permit full control over the probe states used. In this case, it is possible to extend our formalism to determine the optimal probe state. To see how, note that the \textsc{hcrb} is a bi-convex function of the probe state $\rho$ and the observable $X$. We have already determined the optimal observable corresponding to a chosen state: $C_{\mathsf{H},\rho}(X) = C_\mathsf{H}(\rho,X)$. Conversely, fixing $X$ amounts to solving a convex problem in $\rho$ to determine the optimal state corresponding to the choice in $X$: $C_{\mathsf{H},X}(\rho) = C_\mathsf{H}(\rho,X)$. Based on this, we can implement an efficient iterative bi-convex program that alternatively updates the state and optimal observables by fixing one and solving the corresponding convex optimisation problem~\cite{Gorski2007_MMOR}.


\section{Applications}
\label{sec:applications}

\noindent
Quantum metrology has applications in both spin and bosonic systems. We demonstrate the broad applicability of our results by showing how our bounds work in each of the two settings. First, for spin systems, one natural problem to consider is that of estimating the different components of a magnetic field. When the total magnetic field is known, there are only two independent components of a magnetic field to estimate, and such a problem can be tackled directly using our analytical approach for two-parameter estimation. In particular, our simple approach using the Algorithm in Fig.~\ref{algo:Algo_hcrb2_2} already gives interesting insights into the problem of quantum magnetometry on various types of noisy probe states.

Secondly, for bosonic systems a key obstacle in determining the ultimate precision limits on parameter estimation is the infinite dimension of such systems. We show that using our analytical approach, this obstacle can be overcome. Specifically, we use our tight analytical bounds (algorithm~\ref{fig:masteralgo}) to calculate the precision bounds on estimating the incompatible components of a logical Bloch vector of a pure bosonic codestate when mixed with a thermal state.


\subsection{Magnetometry}
\label{subsec:magnetometry}
\noindent
We use our simple two-parameter bounds to consider magnetic field sensing, which has important technological applications in navigation, position tracking, and imaging~\cite{Razzoli2019_PRA}. We apply our method of finding the \textsc{hcrb} to the estimation of a magnetic field $\mathbf{B} = (B_x,B_y,B_z)$ in three dimensions. Quantum magnetometry is an important application of quantum metrology, and is essential for detecting defects and realising compact magnetic resonance imaging scanners~\cite{Liu2019_PRL}. Estimating each component individually allows us to attain the quantum limit~\cite{Paris2009_IJQI}, and this has been demonstrated in several studies~\cite{Pang2014_PRA, Jones2009_S}. However, in many practical applications, knowledge of multiple parameters is required simultaneously, and we must consider joint estimation strategies.

The three parameters of interest $\smash{\bm{\theta} = (\theta_1, \theta_2, \theta_3)^\top}$ appear in the single spin Hamiltonian $\hat{H}_j(\bm{\theta}) = \bm{\theta}\cdot\mathbf{S}_j$, where $\smash{\mathbf{S}_j}$ is the spin operator for the $j^{\text{th}}$ spin.
Local depolarising noise, described by the single spin \textsc{cptp} map 
\begin{align}
 \mathcal{D}_{g}[\rho] = (1-g)\rho + g \frac{\mathbbm 1_2}{2} \, ,
\end{align} 
provides a general description for a noisy environment, where $g$ denotes the depolarisation magnitude and takes values between 0 and 1. The parameters $\bm{\theta}$ are imprinted on the probe state via the unitary evolution $\smash{\hat{U} = \exp[-i\hat{H}_j(\bm{\theta})]}$. For our example, we assume that the magnetic field in the $z$-direction is known, and we therefore wish to estimate the two parameters $B_x$ and $B_y$. We choose an identity weight matrix to equally prioritise each parameter into a weighted scalar mean square error. We consider three families of $n$-spin probe states, namely the traditional \textsc{ghz} states for single-parameter estimation, the modified 3D-\textsc{ghz} states introduced by Baumgratz and Datta~\cite{Baumgratz2016_PRL}, and the \textsc{gnu} states introduced by Ouyang in the context of quantum error correction \cite{Ouyang2014_PRA}.

First, the 3D-\textsc{ghz} state can be written as 
\begin{align}
    \ket{\psi_n^{\text{3D-GHZ}}} = \frac{1}{\mathcal{N}}\sum_{j=1}^3 \left(\ket{\smash{\phi^+_j}}^{\otimes n} + \ket{\smash{\phi^-_j}}^{\otimes n}\right),
    \label{eqn:3d_ghz_states}
\end{align}
where $n$ is the total number of spins, $\mathcal{N}$ is the normalisation constant of the state and $\ket{\smash{\phi^{\pm}_j}}$ are the eigenvectors corresponding to the $\pm 1$ eigenvectors of the $j^{\text{th}}$ spin matrix. The evolved state then becomes 
\begin{align}
 \rho(\bm{\theta}) = \hat{U}(\bm{\theta})\mathcal{D}_{g}^{\otimes n}[\ket{\smash{\psi_n^{\text{3D-\textsc{ghz}}}}}\bra{\smash{\psi_n^{\text{3D-\textsc{ghz}}}}}]\hat{U}(\bm{\theta})^\dagger\, . 
\end{align} 
We illustrate how the upper bound to the \textsc{hcrb} and the \textsc{qcrb} change with the number of probe qubits $n$ for different depolarising channel strengths $g$ in Fig.~\ref{fig:mag_hcrb_bound}. We observe that the upper bound to the \textsc{hcrb} is indeed tighter than the \textsc{qcrb}. Both variance bounds increase with an increasing depolarising probability of the depolarising channel, as expected. The 3D-\textsc{ghz} state attains the Heisenberg precision scaling for the noiseless case.

\begin{figure}[t!]
\centering
\includegraphics[width =\columnwidth]{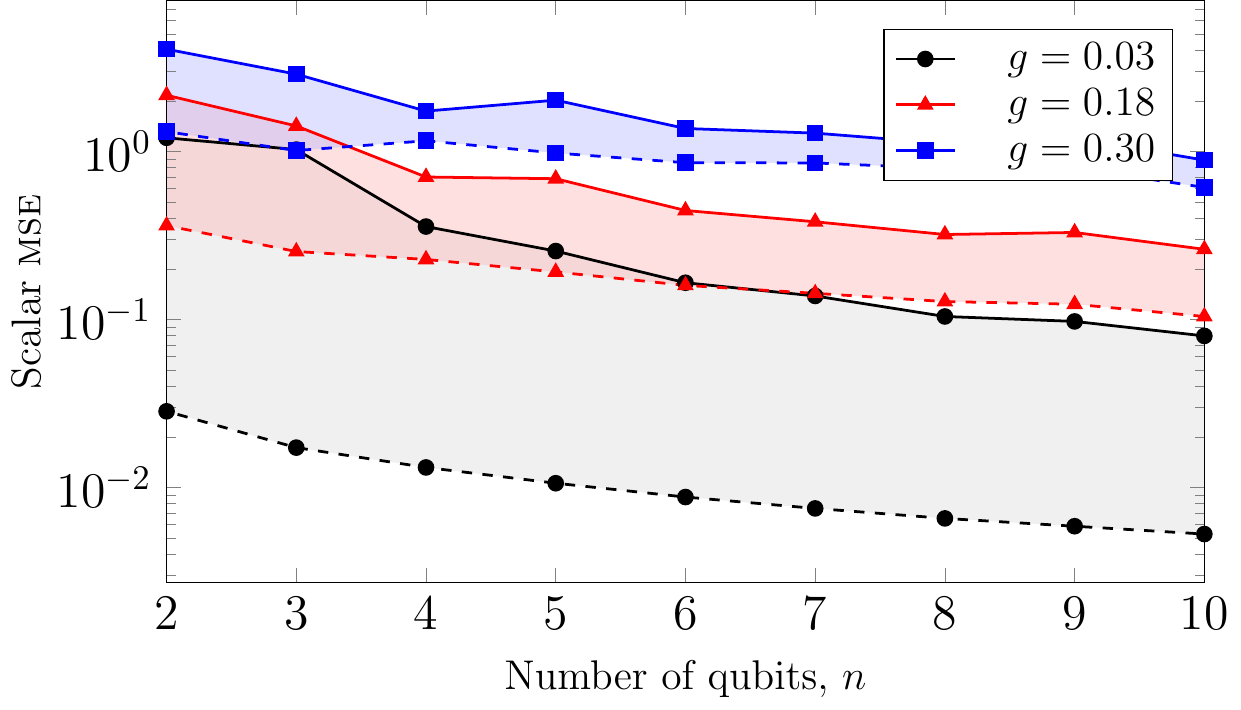}
    \caption{Estimation precision of two directional components of the a magnetic field $B_x$ and $B_y$ with increasing number of spins in a depolarising environment parameterised by $g$, using modified 3D-\textsc{ghz} states. The corresponding Hamiltonian is $ \sum_{j=1}^n \hat H_j({\bm \theta})$ and has no interactions terms. With the identity weight matrix, the dashed lower lines illustrate the \textsc{qcrb} and the solid lines the upper bound $\mathscr{U}$ to the \textsc{hcrb} as given by Theorem \ref{thm:2-param}.}
\label{fig:mag_hcrb_bound}
\end{figure}%
%

Second, we consider the class of \textsc{gnu} states that are robust to a constant amount of erasure and dephasing \cite{Ouyang19_arxiv}:
\begin{align}
|\varphi_1\>   =  \frac{1}{2} \sum_{j=0}^2 \sqrt{\binom 2 j }|D^{n}_{G j}\>, \label{eq:resource-state-general}
\end{align}
Here, for every $w = 0,\dots, n$, the Dicke state $|D^n_w\>$ is a uniform superposition over all computation basis states $|x_1\>\otimes \dots \otimes |x_n\>$ with Hamming weight $w$. Since $n=2G$, where $G$ is related to the number of bit-flip errors that can be corrected, we present results for the \textsc{gnu} states for which $n$ is even. These are shown in Fig.~\ref{fig:holevo_State_comp}, and compared with traditional \textsc{ghz} states and 3D-\textsc{ghz} states. The traditional \textsc{ghz} states give a worse estimation for larger qubit number at constant depolarisation rate, as is well-known. The \textsc{gnu} states perform similarly to the 3D-\textsc{ghz} states.

Finally, we use our formalism to determine the optimal $n$-qubit observables $X_1$ and $X_2$ that attain the \textsc{hcrb} using the 3D-\textsc{ghz} states. Unlike the single qubit estimation case, analytic solutions to these observables are challenging and the dimension of $X_j$ scales as $\smash{2}^n$. Instead, we  numerically determine their structure, shown in Fig.~\ref{pic:grid_deformation_types}. We plot the real and imaginary parts of the matrices $X_1$ and $X_2$, and the Hermiticity of the observables is clearly observed.

\begin{figure}[t!]
\centering
\includegraphics[width =\columnwidth]{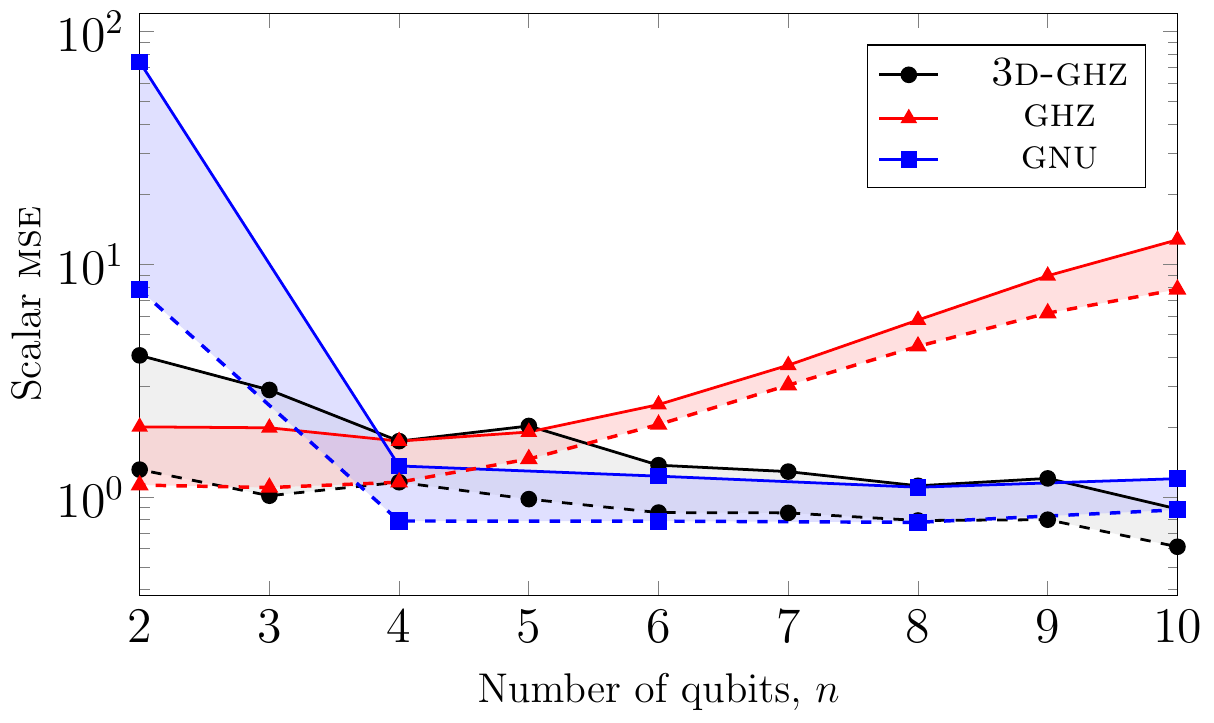}
\caption{Using theorem \ref{thm:2-param}, we depict two-sided bounds on the \textsc{hcrb} for estimating the elements of a magnetic field in a depolarising channel. All plots with depolarising strength $g=0.3$. This plot compares the performance of 3D-\textsc{ghz} states with $n$-qubit \textsc{ghz} states, and the permutation-invariant \textsc{gnu} states. We observe an interesting crossover point between the tightest bound generated by the n-qubit \textsc{ghz} states and 3D-\textsc{ghz} states, with increasing number of qubits. The \textsc{gnu} states are defined over the even number of qubit numbers generates the lowest Holevo bound for small number of qubits.}
\label{fig:holevo_State_comp}
\end{figure}%
%


\subsection{Bosonic quantum codes}
\label{eqn:codes}

\noindent
Our formalism for the \textsc{hcrb} allows performance characterisation of fault tolerant quantum codes in the context of quantum metrology. In this section, we apply our master algorithm in Fig.~\ref{fig:masteralgo} to explore how bosonic error-correcting codes can improve characterisation of logically encoded states in the presence of noise.

Susceptible quantum information can be safeguarded from decoherence by storing it in quantum error-correcting codes (\textsc{qecc}), which in the case of continuous variable (bosonic) quantum systems are subspaces of infinite-dimensional Hilbert space. The working principle of \textsc{qecc} is to project states with errors with high probability onto correctible subspaces labelled by the error syndromes, and dynamically evolve the projected state back to the original code space. 
When these codes are well-chosen, they can correct against errors that are introduced in physically realistic noise models. While bosonic codes on multiple-modes that correct against displacement errors~\cite{GKP01,noh2019-bosonic2bosonic} and photon loss~\cite{CLY97,WaB07,BvL16,ouyang2019permutation} exist, a key attraction of bosonic codes is that they can be used even on a single mode. For example, to protect codes against photon loss and phase errors on a single-mode, one can use codes gapped in the Fock basis~\cite{BinomialCodes2016,GCB20-PhysRevX.10.011058}, or a single-mode \textsc{gkp} code for displacement errors~\cite{GKP01}. For a complete exposition of fault tolerant quantum computing and error correcting codes, the reader is directed to references~\cite{Terhal2015_RMP, Michael2016_PRX,Terhal2020_QST}.

\begin{figure}
\subfloat[$\text{Re}(X_1)$]{\includegraphics[width=0.49\columnwidth]{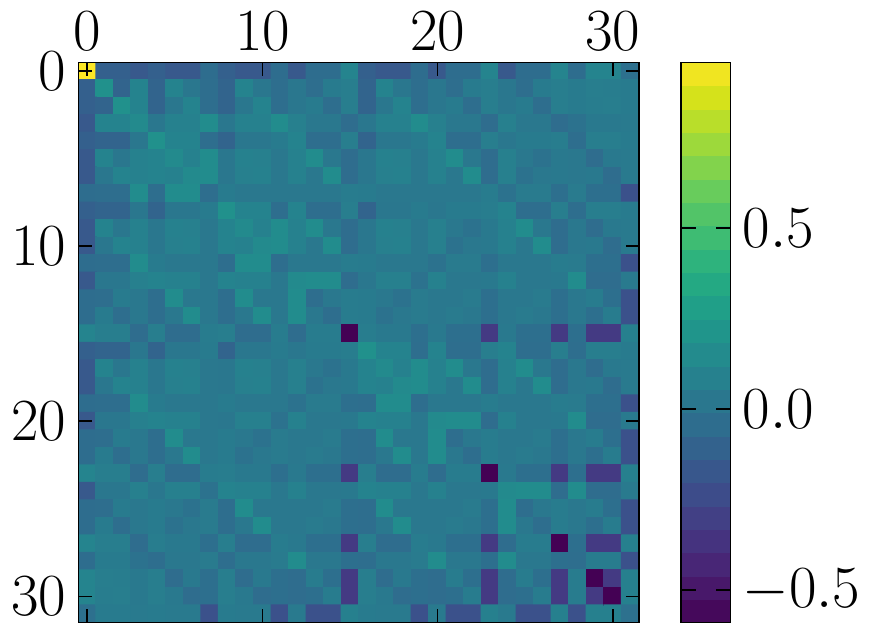}\label{fig:ReX1n5}} 
\subfloat[$\text{Im}(X_1)$]{\includegraphics[width=0.49\columnwidth]{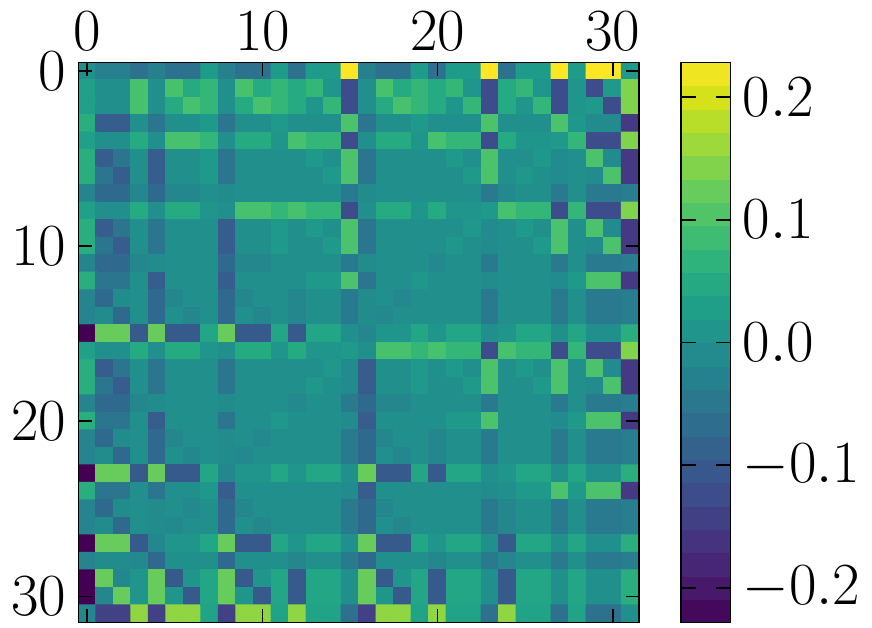} \label{fig:ImX1n5}} \\
\subfloat[$\text{Re}(X_2)$]{\includegraphics[width=0.49\columnwidth]{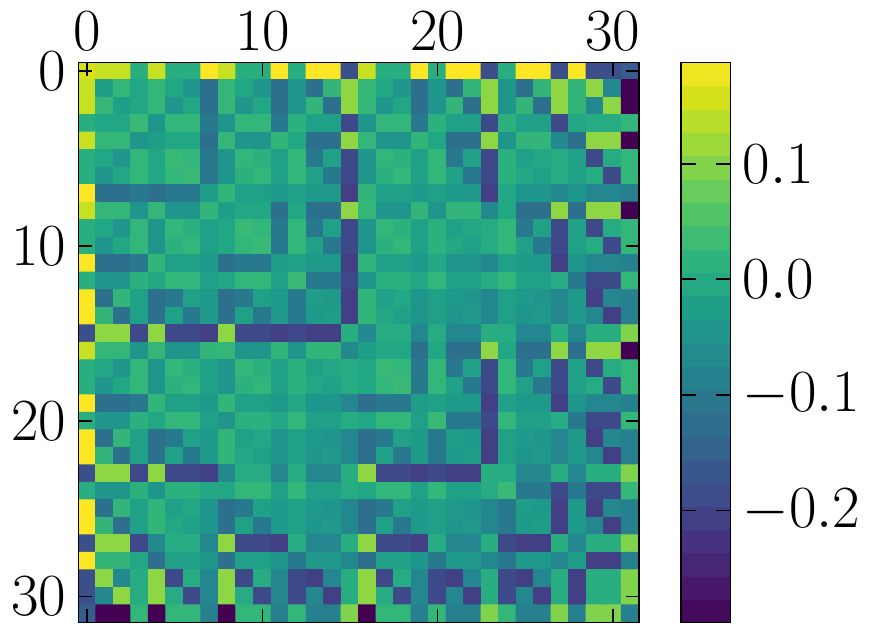} \label{fig:ReX2n5}} 
\subfloat[$\text{Im}(X_2)$]{\includegraphics[width=0.49\columnwidth]{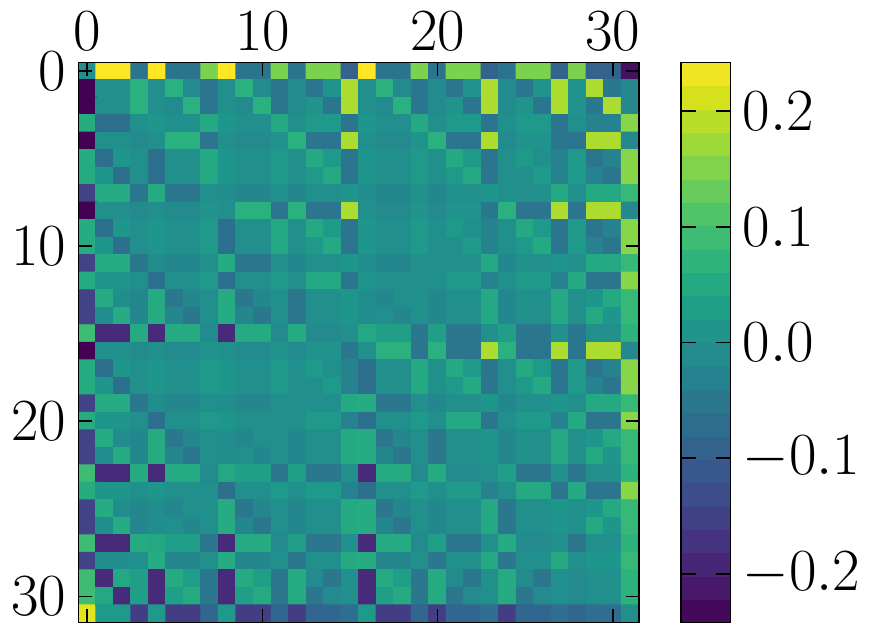} \label{fig:ImX2n5}}
\caption{Using Algorithm 1, we obtain a heatmap for the measurement observables $X_1$ and $X_2$ for the depolarised five-qubit 3D-\textsc{ghz} state, under a depolarising strength $g=0.3$. We plot the real and imaginary parts separately. The Hermiticity of these observables are clearly illustrated.}
\label{pic:grid_deformation_types}
\end{figure}%

We focus on bosonic codes, which, unlike two-level systems, has infinite energy levels per mode. We define a logically encoded state that is parameterised by the coordinates $\bm{\theta} = (x, \phi)^\top$ with $x=\cos(\theta/2)$. We evaluate the \textsc{hcrb} for this bivariate estimation scheme by using theorem~\ref{thrm:hcrb_bounds_arb_u} to evaluate upper and lower bounds to the \textsc{hcrb} for fixed $u$. We then tune the parameters of the binomial codes to effect improvements to estimates of $\bm{\theta}$. We consider binomial codes that protect codewords against number-shift and phase errors. In particular, we analyse the ultimate limits of estimating the complex coefficients of a pure binomial codestate in the presence of thermal noise. 

The logical codewords for the binomial code are supported on a bounded number of Fock states through
\begin{align}
|0_L\> &= \sum_{\substack{ j \ge 0  \\ j \ \even} }
2^{-\frac{(n-1)}{2}}\sqrt{\binom n j} |Gj\>, \\
|1_L\> &= \sum_{\substack{ j \ge 0  \\ j \ \odd} }
2^{-\frac{(n-1)}{2}}\sqrt{\binom n j} |Gj\>,
\end{align}
where $G, n \in \mathbbm{R}$ are related to the number of correctable number-shift and phase errors respectively \footnote{The parameters $G$ and $n$ here are reminiscent of \textsc{Gnu}-states in permutation-invariant quantum codes \cite{Ouyang2014_PRA}}.
From Ref.~\cite[Eq (7)]{BinomialCodes2016}, one requires $G \ge G_{\rm bin}+L_{\rm bin}+1$ to correct $G_{\rm bin}$ gain and $L_{\rm bin}$ loss errors, and $n-1 \ge \max\{L_{\rm bin},G_{\rm bin},2D_{\rm bin}\}$ to correct $D_{\rm bin}$ phase errors. 
For fixed $G$ and $n$, we construct the logical state $\smash{\rho_L = \ket{\psi_L}\bra{\psi_L}}$ with
\begin{align}
    |\psi_L\> =  x |0_L\> + \sqrt{1-x^2}  e^{i \phi}|1_L\>,
\end{align}
where $x \in[-1,1]$ and $\phi \in \mathbb R$. In the noisy scheme, we thermalise this logical pure state through $\smash{\rho = \lambda_{\rm th} \rho_{\rm th} +  (1-\lambda_{\rm th}) \rho_L}$, where
\begin{align}
\rho_{\rm th} = \frac{1}{1-e^{-\beta }} \sum_{k=0}^\infty e^{-k \beta }|k\>\<k| ,    
\end{align}
is a thermal state with temperature $\beta$. Since $\smash{\partial_j\rho = (1-\lambda_{\rm th})\partial_j\rho_L}$, $j=\{x, \phi\}$
and $\rho_L$ is only supported on the Fock states $|0\>, |G\> , \dots, |Gn\>$,
the state derivatives $\partial_j\rho$ are only supported on the Fock states $|0\>, |G\>, \dots ,|Gn\>$.
Using this property, we can determine that in the calculation of the \textsc{hcrb}, we need only consider the evaluation of $\rho$ on the support of the Fock states $|0\>, |G\>, \dots ,|G n\>$.  
Denoting such a state as $\tau$, we can write
\begin{align}
    \tau = \sum_{j,k=0}^n |j\>\<k| \<Gj|\rho|Gk\>,
\end{align}
and observe that its has spectral decomposition
\begin{align}
    \tau = \sum_{k=0}^n t_k |\tau_k\>\<\tau_k| \quad \text{with} \quad |\tau_k\> = \sum_{j=0} ^n \tau_{k,j} |j\>.
\end{align}
The key implication is that $\tau$ is now an effective size $(n+1)$ matrix, and unlike $\rho$, does not have infinite dimensions.
Now define
\begin{align}
    |T_k\> =  \sum_{j=0} ^n \tau_{k,j} |Gj\>.
\end{align}
From the form of our noise model, $\tau$ is a full rank matrix because it is a convex combination of a positive definite matrix, and a positive semi-definite matrix. The positive definite matrix arises from a truncation of the thermal state on the Fock states $|0\>,|G\>,\dots, |Gn\>$, and the positive semi-definite matrix arises from $\rho_L$. Since $\tau$ is a full rank matrix, it follows that the spectral decomposition of $\rho$ is
\begin{align}
\rho =&
\sum_{k=0}^n t_k |T_k\>\<T_k| +
\lambda_{\rm th} \sum_{k=0}^{n-1}
\sum_{j=1}^{g-1} \frac{e^{-\beta(Gk+j)}}{1-e^{-\beta}} |Gk+j\>\<Gk+j|\notag\\
&+ \lambda_{\rm th} \sum_{k=Gn+1}^\infty 
\frac{e^{-\beta k}}{1-e^{-\beta}} |k\>\<k|.\label{terms-in-rho}
\end{align}
From the above spectral decomposition of $\rho$, it is clear that only the first summation term contributes to the state derivatives. This makes the effective dimension of the problem equal to the dimension of $\tau$ instead of that of $\rho$. Because of this reduction in the effective dimensionality of the problem, we can efficiently use algorithm~\ref{fig:masteralgo} to evaluate upper and lower bounds to the \textsc{hcrb} for fixed $u$ to benchmark parameter estimates for $\smash{\bm{\theta}}$. 

\begin{figure*}[t!]
\subfloat[Thermalisation effect.]{\includegraphics[width=0.315\linewidth]{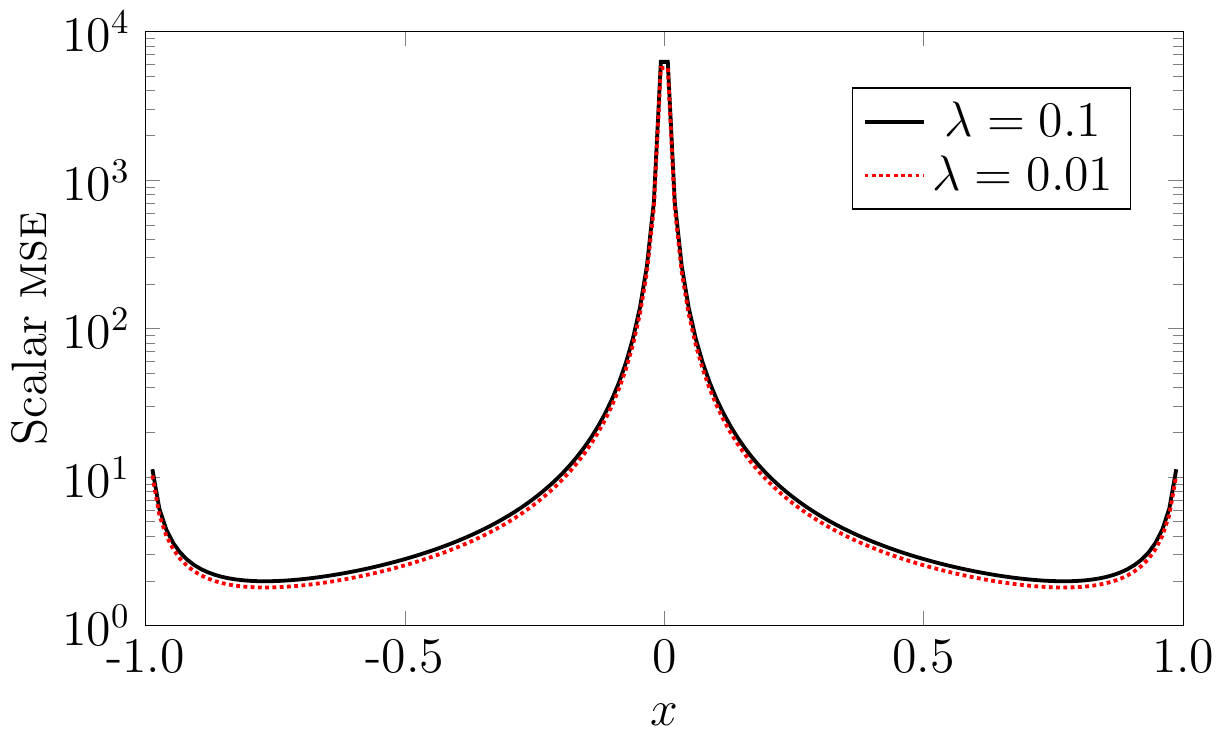}\label{fig:lambda_hcrb}} \hspace{2pt}
\subfloat[Amplitude damping error parameter, $G$.]{\includegraphics[width=0.323\linewidth]{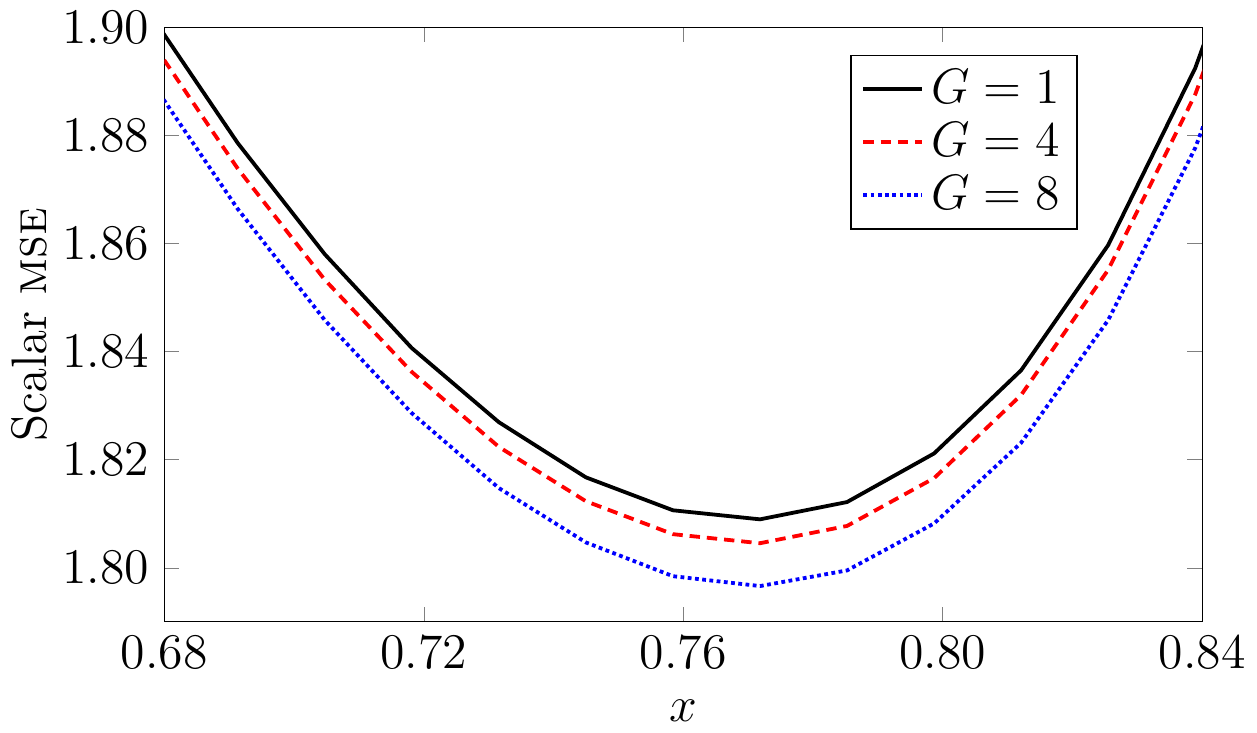} \label{fig:g_hcrb}}\hspace{2pt}
\subfloat[Phase error parameter, $n$]{\includegraphics[width=0.337\linewidth]{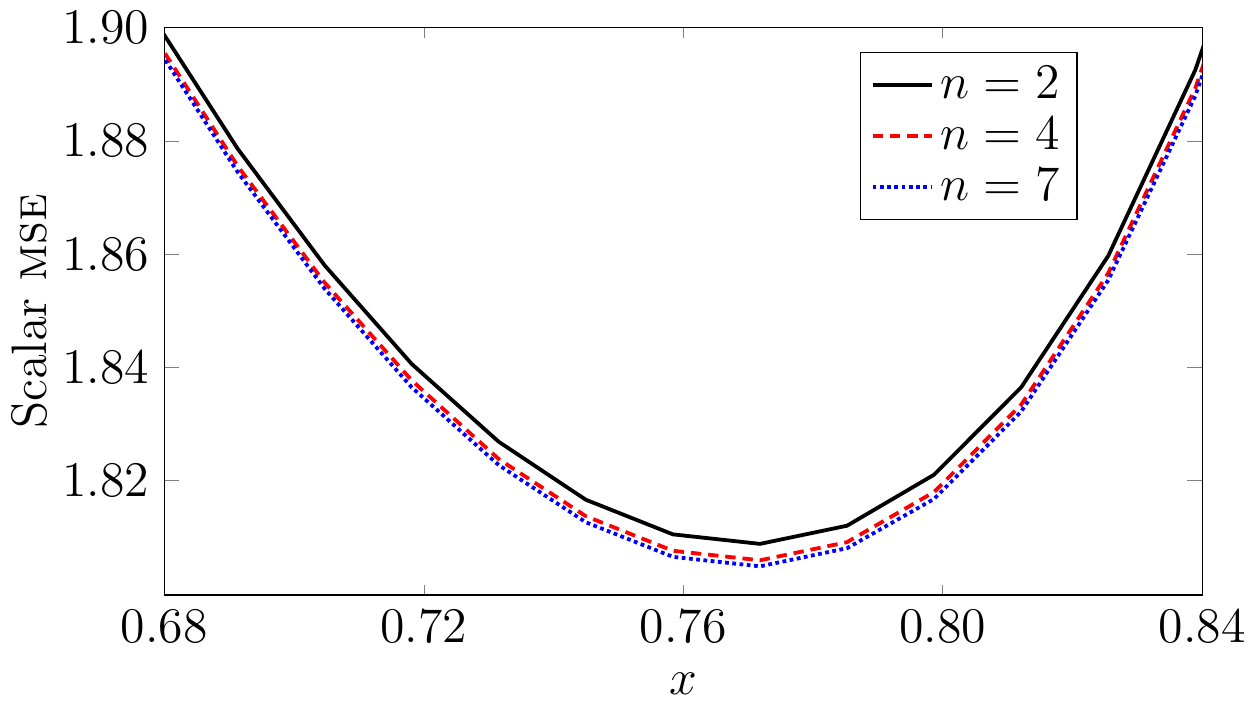} \label{fig:n_hcrb}}
\caption{Variation of \textsc{hcrb} with thermalisation and binomial code parameters. For the plot in Fig.~\ref{fig:lambda_hcrb}, we take $n=2$, $\beta=0.01$, and $G=1$. In Fig.~\ref{fig:g_hcrb}, we see the effect of noise by varying the number of correctable amplitude damping errors of the binomial code, using $n=2$, $\lambda_{\rm th}=0.01$, and $\beta=1$. In Fig.~\ref{fig:n_hcrb} we illustrate the effect of increasing the correctable phase errors of the binomial code on the \textsc{hcrb} using $\lambda_{\rm th}=0.01$, $\beta=1$, and $G=1$.}
\label{pic:holevo_binomial_codes}
\end{figure*}%

By optimising over $u$, we find exact values of the \textsc{hcrb}. In Fig.~\ref{pic:holevo_binomial_codes} we illustrate how the \textsc{hcrb} changes with different values of the noise parameters $\lambda_{\rm th}$ and $\beta$ and the code parameters $G$ and $n$. Fig.~\ref{fig:lambda_hcrb} shows that increasing the thermalisation increases the scalar \textsc{mse}. We also note that the value for $x$ that minimises the \textsc{mse} is insensitive to the amount of thermalisation in the state. Notice also that the effect of binomial codes in this application is limited at high thermalisation, where the state has a lot of thermal noise. The binomial codes are able to protect against errors resulting on the state due to low temperatures. To see this, we observe the behaviour of the \textsc{hcrb} in the region where it is minimised with respect to $x$. In Fig.~\ref{fig:g_hcrb} we illustrate the behaviour of the \textsc{hcrb} within $x= [0.68,0.84]$ for different values of $G$. Notice that increasing the number of correctable amplitude damping errors by increasing the value of $G$ improves the precision of the simultaneous estimate for $x$ and $\phi$. In Fig.~\ref{fig:n_hcrb}, we illustrate a similar improvement by increasing $n$, which is related to the number of correctable phase errors. This demonstrates that the error correcting codes can be used to improve simultaneous parameter estimates in the low error regime.

It is worth noting the performance of the simple bounds in theorem~\ref{thm:2-param} for this application. For the specific choice of parameters $u=0$, $x=-0.8$, $\phi=0.7$, $n=5$, $G=3$, $\lambda_{\rm th}=0.2$, $\beta=0.1$, we get the exact \textsc{hcrb}. In general, theorem~\ref{thm:2-param} returns non-tight bounds to the \textsc{hcrb} for alternative values in the parameter space. This illustrates that if the tightness of bounds are crucial, then one should apply theorem~\ref{thrm:hcrb_bounds_arb_u}.


\section{Conclusions and Discussions} 
\label{sec:conclusions}

\noindent
Quantum metrology promises practical near term quantum technologies. Experimental developments in sensing are demonstrating early theoretical results and advancements in estimation theory. On the theoretical front, one prominent limitation that remains is the estimation of multiple non-compatible observables. Specifically, the optimal strategy to define the fundamental limits to precision estimates and their attainability is not known. Efforts to estimate multiple non-compatible observables have largely been focused on approaching the fundamental quantum Cram{\'e}r-Rao bound (\textsc{qcrb}). This has led to efforts to devise non trivial measurement schemes that approach the \textsc{qcrb}. An alternative approach is to focus on the tighter Holevo Cram{\'e}r-Rao bound (\textsc{hcrb}), which is physically attainable. However, the \textsc{hcrb} is difficult to evaluate since it involves a difficult optimisation over two observables. This has limited its application in quantum estimation theory. 

In this paper, we have made significant progress in analytically solving the \textsc{hcrb} for two-parameter estimation problems, and providing bounds for larger number of parameters. In the two-parameter case, we reduce the complexity of the optimisation procedure to that of solving a set of linear equations, which can be easily solved using most numerical software packages. We also provide analytic expressions for the optimal positive operator valued measurements (\textsc{povm}). Our results readily apply to a large range of physical applications. This will provide deeper insight into the role of quantum measurements in quantum sensing, and help continue the drive of realising quantum technologies.  

We illustrate an application of our results by considering the estimation of a magnetic field using noisy multi-qubit probe states \cite{Ouyang19_arxiv,ouyang2020weight}. A recent numerical study by Albarelli \emph{et al.} demonstrated the necessity of using the \textsc{hcrb} over the \textsc{qcrb}, based on a violation of the weak commutation condition~\cite{Albarelli2019_PRL}. Here, we provide further insight into the role that the \textsc{hcrb} plays in quantum estimation theory. We provide conditions for when this bound is tighter than the \textsc{sld qcrb} (or the Helstrom bound) and provide the corresponding optimal measurement observables. 


A second application of our results explores how bosonic quantum error correction codes can improve noise resilience of parameter estimates. Bosonic codes are interesting because of their potential in reducing the number of physical systems required while having some robustness against errors. However, the infinite dimensionality of bosonic systems renders a brute force numerical approach to determining the \textsc{hcrb} intractable. Instead, through our analytical approach, we reduce this problem to a finite dimensional problem, and thereby evaluate the corresponding precision bounds efficiently.


There are several clear extensions of our work that can be readily addressed. The first would be to use the analytic expressions that we derive to provide further insight into more protocols in estimation theory. We hope that this will help to drive the wave for experimental validation. A second line of work would consider an extension of the Holevo bound to parameters with arbitrary choice of weight matrices. In this work we have considered unit weight matrices, which was motivated through placing equal importance to each parameter. A more general weight matrix would provide a more general bound. A final line of work would consider the optimal implementation of the general measurements that were derived in this work. This would provide an immediate access to the tighter \textsc{hcrb} through experimental implementation.


\section{Acknowledgments}
\label{sec:acknowledgments}

\noindent
J.S.S. and P.K. acknowledge the support of EPSRC via the Quantum Communications Hub through grant number EP/M013472/1. Y.O. and E.C. acknowledge the support of EPSRC through grant number EP/M024261/1. This work was completed and submitted when E.C. was working at the University of Sheffield.


\appendix

\section{Matrix calculus}
\label{app:matrix-calculus}

\noindent
We prove some elementary facts about matrix calculus that we use repeatedly in our analysis of the turning points of the Lagrangian functions that occur throughout the manuscript.

We begin by defining some notations. Since a complex matrix of size $n$ is a map from $\mathbb C^n$ to $\mathbb C^n$, we use $L(\mathbb C^n)$ to denote the set of all complex matrices of size $n$. Here, the notation $L(\mathbb C^n)$ reflects the fact that a matrix is a linear mapping that is an automorphism on $\mathbb C^n$.
At times we are interested in matrices that are also Hermitian, which means that they are equal to their complex conjugates. In this scenario, we use $\mathbb H_n$ to denote the set of all complex matrices that are also Hermitian. Clearly for instance, $\mathbb H_n$ is a strict subset of $L(\mathbb C^n)$.

Now let $f : L (\mathbb C^n) \to \mathbb C$ denote a function that maps a complex matrix to a complex scalar. If $f(Y)$ is differentiable at $Y$ in the direction $H$, we use 
\begin{align}
    \nabla_{Y,H} f(Y) = \lim_{h\to 0} \frac{ f(Y + h H) - f(Y) } {h}
\end{align}
to denote the Fr\'echet derivative of $f(Y)$ in the direction $H$. In the above formula, $h$ is a real infinitesimal parameter. Properties of these Fr\'echet derivatives continues to be an active area of research~\cite{deadman2016}, and they have also been recently used in quantum information theory~\cite{ouyang2019mems}. 

In this paper, we are interested in matrix functions that are either linear or quadratic in the matrix variable $Y$. This leads us to analyse the Fr\'echet derivatives given by the following lemma.

\begin{lemma}
\label{lem:Fr\'echet-derivatives}
Let $Y,H \in L(\mathbb C^n)$.
Then 
\begin{align}
 \nabla_{Y,H} \tr{A Y} &= \tr{A H} \\
 \nabla_{Y,H} \tr{A Y^\dagger} &= \tr{A H^\dagger} \\
 \nabla_{Y,H} \tr{Y A Y^\dagger} &= \tr{A Y^\dagger H + Y A H^\dagger} \\
 \nabla_{Y,H} \tr{Y^\dagger A Y} &= \tr{Y^\dagger A H + A Y H^\dagger} .
\end{align}
\end{lemma}
\begin{proof}
The proof of the above results from direct application of the definition of the Fr\'echet derivative for the first two equations. For the last two equations, we also use the cyclic property of the trace.
\end{proof}

We are often faced with the unconstrained minimisation of a quadratic form, and we show in the following lemma what the optimal solution to these optimisation problems are.

\begin{lemma}
\label{lem:matrix-minimize}
Let $A \in L(\mathbb C^n)$ and let $\rho$ be a full rank matrix in $\mathbb H_n$.
Then 
\begin{align}
\min_{Y \in L(\mathbb C^n)} \left( \tr{Y \rho Y^\dagger} + \tr{A Y} + \tr{A^\dagger Y^\dagger}   \right)
&= - \tr{A^\dagger \rho^{-1} A}\label{eq:lem-min-1} \\
\min_{Y \in L(\mathbb C^n)} \left( \tr{Y^\dagger \rho Y} + \tr{A Y} + \tr{A^\dagger Y^\dagger}   \right)
&= - \tr{A \rho^{-1} A^\dagger} ,\label{eq:lem-min-2} 
\end{align}
with the minimum achieved by setting $Y =- A^\dagger \rho^{-1}$ and $Y = -\rho^{-1} A^\dagger$ respectively. 
\end{lemma}
\begin{proof}
We first prove \eqref{eq:lem-min-1} and \eqref{eq:lem-min-2}. The corresponding objective functions that are to be minimised are convex and differentiable, so it suffices to find when their Fr\'echet derivatives are equal to zero for any direction $H$. For this, we use Lemma \ref{lem:Fr\'echet-derivatives}, from which we find that we must have $\rho Y^\dagger+ A = 0$ and $Y^\dagger \rho + A = 0$ respectively. Making use of the fact that $\rho$ is invertible whenever it has full rank, we multiply both sides of the equations, and find that the optimal $Y$s are given by 
$Y = -A^\dagger \rho^{-1}$ and 
$Y = -\rho^{-1} A^\dagger$ respectively. Substituting this back into the objective functions gives the result.
\end{proof}

\section{Simple two-parameter bounds to the \textsc{hcrb}}
\label{app:2_param}

\noindent
We explicitly derive the \textsc{hcrb} for the two-parameter case. In the two-parameter setting, the \textsc{hcrb} with a weight matrix $W$ is given by the optimisation problem
\begin{mini} 
{X_1,X_2}{\tr{ W{\rm Re}Z} + \| W {\rm Im} {Z} \|_1  ,}
{}{}
\addConstraint{\Tr{\rho X_j}}{=0}
\addConstraint{\Tr{\partial_j \rho X_k}}{=\delta_{jk},}
\label{app:opt0}
\end{mini}
where $X_j$ are constrained to be Hermitian matrices in $\mathbb H_N$,  
and $Z$ is a matrix given by
\begin{align}
    Z  = 
    \begin{pmatrix}
    \tr{ \rho X_1^2 } & \tr{ \rho X_1 X_2 }  \\
    \tr{ \rho X_2 X_1 } & \tr{ \rho X_2^2 }  \\
    \end{pmatrix}.
\end{align}
Note that $W$ is always taken to be a positive definite matrix.
For simplicity, we only consider the scenario where $W$ is the identity matrix.
\subsection{Reformulation of the optimisation problem}\label{subsec:2param-reparam}

\noindent
The optimisation problem \eqref{opt0} can be solved analytically primarily from our ability to rewrite the objective function as a quadratic function in the optimisation variables $X_1$ and $X_2$.
The method of Lagrange multipliers when applied to problems with quadratic objective functions and linear equality constraints is well-known to be exactly solvable, for example in theory of portfolio optimisation in finance \cite{Best2010_book}. 
A similar argument will allow us to solve \eqref{opt0} using this method.

We begin by showing why the objective function is quadratic. 
To see this, we first note that that the diagonal terms of $Z$ are positive numbers, because $X_1$ and $X_2$ are Hermitian and $X(\cdot)X^\dagger$ is a completely positive map. 
Second, the positivity of the diagonal entries of $Z$ implies that 
\[\rm{Re} \tr{Z}=\tr{Z} = \tr{X_1 \rho X_1^\dagger} +  \tr{X_2 \rho X_2^\dagger}.\]
Third, the positivity of the diagonal entries of $Z$ implies that the trace norm of ${\rm Im} Z$ can be explicitly evaluated. This is because
the diagonal entries of ${\rm Im} Z$ must be zero.
Since $X_1,X_2$ and $\rho$ are Hermitian matrices,
it follows that 
\[{\rm Im} Z = \frac{1}{2i} \begin{pmatrix}
0 &  w \\
-w & 0 \\
\end{pmatrix},\]
where $w = \tr{\rho X_1 X_2 }  - \tr{X_2 X_1 \rho  }$ is an imaginary number.
The eigenvalues of ${\rm Im}Z$ are therefore $\pm w/2$, which implies that the trace norm of ${\rm Im} Z$ is $\max\{iw,-iw\}$.
From this, 
we get
\begin{align}
    {\rm Re} \tr{Z} +  i w 
=&\Tr{ (X_1 + i X_2) \rho (X_1 + i X_2)^\dagger }\\
     {\rm Re} \tr{Z}    -i w
=&\Tr{ (X_1  - i  X_2) \rho ( X_1 - i X_2)^\dagger }.
\end{align}
Now let us make the substitution
$Y = X_1 + i X_2$.
In this scenario, we can rewrite the equality constraints in \eqref{app:opt0} as 
\begin{align}
\tr{ \rho Y} &= 0 \notag\\
\tr{ \partial_1 \rho Y} &= 1 \notag\\
\tr{ \partial_2 \rho Y} &= i .
\end{align}
Hence the optimisation problem \eqref{opt0} can be written as 
\begin{mini} 
{Y,t}{ t,}
{}{}
\addConstraint{\Tr{ Y \rho Y^\dagger } }{ \le t }
\addConstraint{\Tr{ Y^\dagger \rho Y } }{ \le t }
\addConstraint{ \tr{ \rho Y} }{=0}
\addConstraint{ \tr{ \partial_1  \rho Y} }{=1}
\addConstraint{ \tr{ \partial_2  \rho Y} }{=i.} 
\label{app:opt1}
\end{mini}
Note that the optimisation problem \eqref{app:opt1} is a linear optimisation problem with convex quadratic and linear constraints. 
When the equality constraints are satisfied, the quadratic terms in the inequality constraints are non-negative, and by setting $t$ to be arbitrarily large, we can see that the inequality constraints in \eqref{app:opt1} can always be strictly satisfied. Since \eqref{app:opt1} is also a convex optimisation problem because of its linear objective function and convex constraint functions, the Slater constraint qualification holds with respect to \eqref{app:opt1}. This implies that the first order Karush-Kuhn-Tucker (\textsc{kkt}) conditions suffices to determine the optimality of \eqref{app:opt1}.
\subsection{Analysing the Lagrangian} \label{subsec:2param-lagrangian}

\noindent
The \textsc{kkt} conditions are stated in terms of the Lagrangian of \eqref{app:opt1}.
The column vector of Lagrange multipliers corresponding to the equality constraints is 
\begin{align}
    {\bf z} &= (z_1, z_2, z_3, z_4, z_5, z_6).
\end{align}
The Lagrangian of \eqref{app:opt1} is
\begin{align}
\mathsf{L}(Y,t,u,v,{\bf z})=&
\; t + 
u \tr{ Y \rho Y^\dagger } - u t + v \tr{ Y^\dagger \rho Y } - v t\notag\\
 &+ z_{1} {\rm Re}\tr{\rho Y} 
  + z_{4} {\rm Im} \tr{\rho Y}  \notag\\
&  + z_2 \left( {\rm Re}\tr{\partial_1 \rho Y} -1 \right)
   + z_5 {\rm Im} \tr{\partial_1 \rho Y} \notag\\
&  + z_3 {\rm Re}\tr{\partial_2 \rho Y} 
   + z_6 \left( {\rm Im} \tr{\partial_2 \rho Y} - 1\right),
\label{eqn:Lagr}
\end{align}
where $u,v$ are non-negative Lagrange multipliers corresponding to the inequality constraints.

There are four types of \textsc{kkt} conditions. First is the stationarity of the derivative of the Lagrangian with respect to the primal variables. 
Second is complementary slackness, which states that the product of the constraint functions~\footnote{A constraint function is $f(x)$ written in the canonical form $f(x)=0$ or $f(x) \le 0$} and their corresponding Lagrange multipliers is always zero. 
Third is the feasibility of the primal variables, and fourth is feasibility of the dual variables.
If these \textsc{kkt} optimality conditions hold, then we can obtain the optimal solution and value of the corresponding optimisation problem. 

Now we use the fact that
\begin{align}
{\rm Re}\tr{\rho Y}  &= \frac{\tr{\rho Y} + \tr{\rho Y^\dagger}}{2} \\
{\rm Im}{\tr{\rho Y} } &= \frac{\tr{\rho Y} - \tr{\rho Y^\dagger}}{2i} \\
{\rm Re}\tr{\partial_j \rho Y}  &= \frac{\tr{\partial_j\rho Y} + \tr{\partial_j\rho Y^\dagger}}{2} \\
{\rm Im}{\tr{\partial_j \rho Y} } &= \frac{\tr{\partial_j\rho Y} - \tr{\partial_j\rho Y^\dagger}}{2i}
.
\end{align}
Using this, it follows that
\begin{align}
\mathsf{L}(Y,t,u,v,{\bf z}) =&
t(1-u-v) - {\bf b}^\top {\bf z} + u \tr{Y \rho Y^\dagger}
\notag\\
 & + v \tr{Y^\dagger \rho Y}+\tr{A Y} + \tr{A^\dagger Y^\dagger},
 \end{align}
 where ${\bf b}$ is the column vector $(0,1,0,0,0,1)$ and 
 \begin{align}
 A = z_1 A_1 + z_2 A_2+ z_3 A_3 + z_4 A_4 + z_5 A_5  + z_6 A_6,
 \end{align}
 where
 \begin{align}
 A_1 =  \frac{ \rho}{2} , \quad A_2 = \frac{\partial_1 \rho}{2}, \quad A_3 = \frac{\partial_2 \rho}{2}
\end{align}
and $\smash{\{A_4, A_5, A_6\} = -i\{A_1, A_2, A_3\}}$.

Before we proceed to derive the Lagrange dual function, we note the following.
\begin{enumerate}
\item We prove that the optimal $t$ must be strictly positive from the positive definiteness of $\rho$.
From the positive definiteness of $\rho$, $t$ is equal to zero if and only if $Y$ is 0, but this would violate the feasibility constraints. Hence $t$ cannot be equal to zero.
\item The stationarity \textsc{kkt} condition requires that the derivative of Lagrangian in Eq.~\eqref{eqn:Lagr} be zero with respect to $t$. From this, we observe that the optimal dual variables must satisfy $u+v=1$.
\item The \textsc{kkt} conditions require that the complementary slackness conditions hold for the inequality constraints in \eqref{app:opt1}.
This means that 
\begin{align}
u \left( \tr{Y \rho Y^ \dagger }  - t\right) = 0 \notag\\
v \left( \tr{Y^\dagger \rho Y }  - t\right) = 0 \label{eq:complementary slackness}
\end{align}
If $\tr{Y \rho Y^ \dagger} \neq \tr{Y^\dagger \rho Y }$, exactly one of the constraints corresponding to $u$ and $v$ must be tight, and complementary slackness implies that the optimal $(u,v)$ must be either $(u,v)=(1,0)$ or $(u,v)=(0,1)$.
This corresponds to the scenario where the \textsc{qcrb} is not equal to the \textsc{hcrb}.
If $\|{\rm Im}Z\|_1 = 0$, $(u,v)=(1,0)$ and $(u,v)=(0,1)$ do not necessarily optimize the value of the Lagrange dual, and in general provide a lower bound to the Lagrange dual. 

However if $\tr{Y \rho Y^ \dagger} = \tr{Y^\dagger \rho Y }$, then the ansatzes $(u,v)=(0,1)$ and $(u,v)=(0,1)$ will not yield tight bounds, because complementary slackness will not further constrain the optimal values of $u$ and $v$.
\end{enumerate}

\subsection{Deriving the Lagrange dual functions}\label{subsec:lagrange-duals}

\noindent
When $(u,v)=(1,0)$, the Lagrangian evaluates to 
\begin{align}
\mathsf{L}(Y,t,1,0,{\bf z}) =&
-  {\bf b}^\top  {\bf z}+ 
 \tr{Y \rho Y^\dagger} 
 +\tr{A Y} + \tr{A^\dagger Y^\dagger},
\end{align}
where $ {\bf b} = (0,1,0,0,0,1)^\top$.
Since $\rho$ is full rank, $\rho$ is invertible. Using Lemma \ref{lem:matrix-minimize}, the above is minimised with respect to $Y$ when
$Y = - A^\dagger \rho^{-1}$
with optimal value
$-\tr{ A^\dagger \rho^{-1} A }.$
In this scenario, the Lagrange dual function of \eqref{app:opt1} evaluated with $(u,v)=(1,0)$ is 
\begin{align} 
g(1,0,{\bf z}) = - \tr{ A^\dagger \rho^{-1} A  }
-  {\bf b}^\top  {\bf z}.
\label{dual-1,0}
\end{align}
Similarly when $(u,v)=(0,1)$, the Lagrangian evaluates to
\begin{align}
\mathsf{L}(Y,t,0,1,{\bf z}) &= 
-  {\bf b}^\top  {\bf z}+ 
 \tr{Y^\dagger \rho Y } +\tr{A Y} + \tr{A^\dagger Y^\dagger},
\end{align}
and is minimised when
$Y = -\rho^{-1} A^\dagger$
with an optimal value of 
$-\tr{ A \rho^{-1} A^\dagger }.$
In this scenario, the Lagrange dual function of \eqref{app:opt1} evaluated with $(u,v)=(0,1)$ is 
\begin{align} 
g(0,1, {\bf z}) = - \tr{ A \rho^{-1} A^\dagger }  
-  {\bf b}^\top  {\bf z}.
\label{dual-0,1}
\end{align}
The Lagrange dual functions 
$g(1,0,{\bf z})$ and $g(0,1,{\bf z})$
can be rewritten in terms of 
the matrices $Q_1$ and $Q_2$ where in the Dirac bra-ket notation, we have
\begin{align}
    Q_1 &= \sum_{j,k=1,\dots,6} \tr{ A_j ^\dagger \rho^{-1}  A_k} |j\>\<k|,\notag\\
    Q_2 &= \sum_{j,k=1,\dots,6} \tr{ A_j \rho^{-1}  A_k ^\dagger} \ket{j}\bra{k}.
    \label{eqn:qmatdefoss}
\end{align}
Here $\ket{j}$ denotes a column vector and $\bra{k}$ denotes a row vector.
The Lagrange dual function that we consider are thus
\begin{align}
    g(1,0, {\bf z}) &= - {\bf z}^\top Q_1  {\bf z} - {\bf b}^\top  {\bf z}, \notag\\
    g(0,1, {\bf z}) &= - {\bf z}^\top Q_2  {\bf z} - {\bf b}^\top  {\bf z} .\label{eq:lagrange-dual-functions}
\end{align}
\subsection{Upper and lower bounds}
\noindent
Using the fact that the Lagrange dual functions  \eqref{eq:lagrange-dual-functions}
and the dual variables are real, 
lower bounds to the \textsc{hcrb} are given by
\begin{align}
    \max_{j=1,2} 
    \max_{{\bf z}\in\mathbb R^6} \left(-{\bf z}^\top {\rm Re}(Q_j) {\bf z} - {\bf b}^\top {\bf z}\right)  ,\label{dual-optimizations}
\end{align}
from which it follows that 
\[
2{\rm Re}( Q_j ){\bf z} + {\bf b} = 0 
\]
is the correct optimality condition to consider.
Thus, when ${\rm Re}(Q_j)$ is full rank, the lower bounds to the \textsc{hcrb} can be written as
\begin{align}
   \max_{j=1,2} l_j, \quad \text{where}\quad l_j = \frac{1}{4}{\bf b}^\top {\rm Re}(Q_j)^{-1} {\bf b}
\end{align}
Interestingly, when $\rho$ is full rank, the matrices ${\rm Re}(Q_j)$ are also full rank. We demonstrate this in the next subsection.

To obtain upper bounds to the \textsc{hcrb}, we appeal to the form of the primal problem is closely related to \eqref{app:opt1}, that has an objective function of 
\begin{align}
\max \{ \tr{Y \rho Y^\dagger } ,  \tr{Y^\dagger \rho Y }\} 
    \label{app:primal-objective-function}.
\end{align}
The upper bounds will be expressed in terms of the dual variables that optimise \eqref{dual-optimizations}, which we can write as ${\bf z}_1 = (z_{1,1},\dots, z_{6,1})$
and 
 ${\bf z}_2 = (z_{1,2},\dots, z_{6,2})$
 where
\begin{align}
    z_{a,j} = 
    -\frac{1}{2}([{\rm Re}(Q_j)^{-1}]_{a2}  +
    [{\rm Re}(Q_j)^{-1}]_{a6}).
\end{align}
and $[{\rm Re}(Q_j)^{-1}]_{ab}$ denotes the matrix element in the $a$th row and $b$th column of the inverse of ${\rm Re}(Q_j)$.

Recall that when $(u,v)=(1,0)$, the optimal solution to $Y$ in minimising the Lagrangian is $-A^\dagger \rho^{-1}$. By choosing $A=A_1 z_{1,1}+ \dots + A_6 z_{6,1}$,
we find that the \textsc{hcrb} is thus upper bounded by
\begin{align}
    P_1 = \max \left\{ \tr{A^\dagger \rho^{-1} A } ,  \tr{ \rho^{-2} A \rho A^\dagger}  \right\}.
\end{align}
Then,
\begin{align}
    P_1 = \max \{ l_1,m_1\}    ,
\end{align}
where
\begin{align}
    m_j = \sum_{a,b=1}^6
    \tr{ \rho^{-2} A_a \rho A_b^\dagger  }
   z_{a,j}z_{b,j}.
\end{align}
When $(u,v)=(0,1)$, the optimal solution to $Y$ in minimising the Lagrangian is $- \rho^{-1}A^\dagger$. In this case the primal objective function is equal to 
\begin{align}
    P_2 = \max \{ l_2,m_2 \} .
\end{align}
Hence the \textsc{hcrb} is at most 
\begin{align}
\min_{j=1,2}\left\{ \max \{l_j,m_j \} \right\}.
\end{align}
This concludes the proof of theorem~\ref{thm:2-param} for bounds on the two-parameter \textsc{hcrb}.

\subsection{Full-rankness of Q}
\label{subsec:full_rank_of_Q}
\noindent
The analytic solution to the \textsc{hcrb} requires ${\rm Re}(Q_j)$ to have full rank such that the solution can be determined. In this subsection, we demonstrate that the full-rankness of the probe state $\rho$ entails the full-rankness of these matrices. Since the regularity conditions of estimation theory require the state to be full-rank, our solution to the \textsc{hcrb} always exists. 

Notice that the matrices $Q_j$ defined in Eq.~\eqref{eqn:qmatdefoss} can be written
\begin{align}
	Q_{1} =  \begin{pmatrix} H &-iH  \\
	iH & H
	\end{pmatrix}, \quad	Q_{2} = \begin{pmatrix} H &iH  \\
	-iH & H
\end{pmatrix},
\end{align}
where $H$ is Gram matrix defined as follows.  We consider the Hilbert-Schmidt inner-product $\langle X,Y \rangle=\mathrm{Tr}[X^\dagger Y]$.  We define the operators 
\begin{align}
		B_1 & = \rho^{-1/2} A_1 = \rho^{-1/2} \rho / 2 = \rho^{1/2} / 2 , \\ 
		B_2 & = \rho^{-1/2} A_2 =   \rho^{-1/2} \delta_1 \rho / 2, \\ 
		B_3 & = \rho^{-1/2} A_3 = \rho^{-1/2} \delta_2 \rho / 2 .
\end{align}	
Then, we have that $H$ is a Gram matrix with respect to this set of operators
\begin{equation}
    H_{i,j} = \langle B_i, B_j \rangle .
 \end{equation}   
As a Gram matrix, it is positive semi-definite. Furthermore, we know that $H$ will be full-rank if and only if the set $\{ B_1, B_2, B_3 \}$ is linearly independent.  We note that $A_1$ cannot be written as a sum of $A_2$ and $A_3$  (since $A_1$ has nonzero trace whereas $A_2$ and $A_3$ are traceless).  Also by a trace-argument, if $\{ A_1, A_2, A_3 \}$ are linearly dependent, we must have that $A_2$ is proportional to $A_3$.  But if $A_2$ and $A_3$ are proportional, then it is really a one-parameter problem and not a two-parameter problem. Hence, $\{ A_1, A_2, A_3 \}$ are linearly independent.  If we assume $\rho$ is full-rank, then $\{ B_1, B_2, B_3 \}=\rho^{-1/2} \{ A_1, A_2, A_3 \}$ is also a linearly independent set.  So full-rankness of $\rho$ entails full-rankness of $H$.  
 
We are now interested in the real part of the $Q_j$ matrices, looking for solutions of $2 {\rm Re}(Q_j) z + b = 0$.  Considering $j=2$
\begin{align}
 	{\rm Re}(Q_{2}) = \left(  \begin{array}{cc} {\rm Re}(H) & - {\rm Im}(H)  \\
 		{\rm Im}(H) & {\rm Re}(H)
 	\end{array} \right) ,
\end{align}   
By performing elementary row operations by taking a linear combination of rows, followed by elementary column operators by taking a linear combination of columns, we get
\begin{align}
 	{\rm Re}(Q_{2}) \rightarrow \left(  \begin{array}{cc} H & iH  \\
 		{\rm Im}(H) & {\rm Re}(H)
 	\end{array} \right) \rightarrow \left(  \begin{array}{cc} H & 0  \\
 		{\rm Im}(H) & H^*
 	\end{array} \right) ,
\end{align}  
where we used ${\rm Re}(H) + i{\rm Im}(H) = H$. Since both rows are linearly independent, ${\rm Re}(Q_j)$ is also always full-rank. Therefore, we have that if the state is full-rank, then so too is the matrix ${\rm Re}(Q_{j})$.

\section{Lower bound in the multi-parameter setting}
\label{app:lower_bounds}

\noindent
By restricting ourselves to the identity weight matrix,
recall that the \textsc{hcrb} is the optimal value of the following optimisation problem over the Hermitian matrices $X_j$ in $\mathbb H_N$ given by
\begin{mini} 
{X_1,\dots ,X_d}{
\tr{ {\rm Re}Z} + \|  {\rm Im} {Z} \|_1  ,}
{}{}
\addConstraint{\Tr{\rho X_j}}{=0}
\addConstraint{\Tr{\partial_j \rho X_k}}{=\delta_{jk}.}
\label{opt0-d}
\end{mini}
For $j,k=\{1,\dots,d\}$, let
 \begin{align}
     w_{j,k} =  
        \tr{\rho X_j X_k} - \tr{\rho X_k X_j}.
        \label{eq:wjk-definition-appendix}
 \end{align}
 
 \subsection{Deriving a lower bound for the objective function}
 \noindent
 In general for a $d$-parameter estimation problem, we have 
 \begin{align}
     \tr{ {\rm Re} Z  } &= 
     \sum_{j=1}^d \tr{X_j \rho X_j^\dagger }  
     \label{ReZ-d}
     \\
     {\rm Im}Z  &= \frac{1}{2i}
     \sum_{\substack{ 1 \le j < k \le d} }
     w_{j,k}( |j\>\<k| - |k\>\<j| ).
     \label{ImZ-d}
 \end{align}
Note that since $Z$ is a Hermitian matrix, ${\rm Im}Z$ is always a skew-Hermitian matrix. For example, when $d=3$, we have
 \begin{align}
     {\rm Im}Z 
     &= \frac{1}{2i}
     \begin{pmatrix}
     0 & w_{1,2} &   w_{1,3} \\
     -w_{1,2} & 0 &  w_{2,3} \\
     -w_{1,3} &  -w_{2,3} & 0\\
     \end{pmatrix}.
 \end{align}
Whenever $d \ge 3$, the trace norm of ${\rm Im}Z$ fails to be a quadratic form in the observables $X_1,\dots, X_d$. Hence, the objective function of the optimisation problem~\ref{opt0-d} fails to be quadratic for us to apply the techniques in Appendix~\ref{subsec:2param-reparam}. 
We can however obtain lower bounds for the trace norm of ${\rm Im}Z$ that do have a quadratic structure, namely by exploiting the following decomposition for the trace norm,  
 \begin{align}
     \| {\rm Im}Z \|_1 = \max\{ \tr  { U {\rm Im}Z   } :   U \mbox{ is a unitary matrix}\}.
 \end{align}
Fortunately, it is possible to pick unitary matrices $U$ such that $\tr  { U {\rm Im}Z}$ are quadratic in the observables $X_1 , \dots , X_d$, which we prove in the following subsection.
We achieve this by constructing unitary matrices labeled by binary vectors ${\bm \alpha} = (\alpha_1, \dots, \alpha_d)$ given by
\begin{align}
    U_{\bm \alpha} 
    = 
    |d\>\<1|(-1)^{\alpha_d} + 
    \sum_{j=1}^{d-1}
    |j\>\<j+1| (-1)^{\alpha_j}.
    \label{Ualpha}
\end{align}
Using the unitary matrices in  
\begin{align}
\mathcal U
=
\left\{
U_{\bm \alpha}
    : {\bm \alpha} \in \{0,1\}^d 
\right\},
\end{align}
we obtain the lower bound
\begin{align}
    \trnorm{ {\rm Im} Z } \ge \max\{  \tr{U {\rm Im} Z  } : U \in \mathcal U \}.
\end{align} 

 \subsection{Recasting the optimisation problem}
 \noindent
In this subsection, we prove the following lemma.
\begin{lemma}
Let $d$ be a positive integer where $d \ge 2$. Now, given a binary vector ${\bm \alpha} = (\alpha_1, \dots, \alpha_d)$,
 let $U_{\bm \alpha}$ be as defined in \eqref{Ualpha}, and let us define 
\begin{align}
V_{\bm \alpha}
    =&
    \frac 1 2 
\tr{( X_d + (-1)^{\alpha_1} i X_{1}) 
\rho ( X_d + (-1)^{\alpha_1} i X_{1})^\dagger}\notag\\
&+
\frac 1 2 \sum_{j=1}^{d-1}
\tr{( X_j + (-1)^{\alpha_j} i X_{j+1}) 
\rho ( X_j + (-1)^{\alpha_j} i X_{j+1})^\dagger},
\label{eq:Vt-defi}
\end{align}
where $X_1,\dots, X_d$ are Hermitian matrices, and $\rho$ is a density matrix. Let
$Z = \sum_{j,k=1}^d\tr{\rho X_j X_k}|j\>\<k| $.
Then 
\begin{align}
\tr{ {\rm Re}Z}
+
\tr{ U_{\bm \alpha} {\rm Im} Z  }
=
V_{\bm \alpha}.
\end{align}
\end{lemma}
\begin{proof}
Rewriting \eqref{ImZ-d}, we get
\begin{align}
    {\rm Im} Z =
    \frac{1}{2i}
    \sum_{j,k=1}^d w_{j,k} |j\>\<k|\label{ImZ-version2},
\end{align}
where $w_{j,k}$ is as defined in \eqref{eq:wjk-definition-appendix}. 
Using \eqref{ReZ-d} and \eqref{ImZ-version2}, we can find that 
\begin{align}
&\tr{ {\rm Re}Z}
+
\tr{U_{\bm \alpha} {\rm Im} Z  }
\notag\\
=&
\sum_{j=1}^d
\tr{X_j \rho X_j^\dagger}
+
\frac{1}{2i}\sum_{j=1}^d \sum_{k=1}^d
\tr{w_{j,k}
U_{\bm \alpha}  |j\>\<k|    }.\label{lem:eq1}
\end{align}
Now 
\begin{align}
    &\tr{ U_{\bm \alpha} |j\>\<k|  }\notag\\
    =&
    \tr{ (|d\>\<1|) |j\>\<k|  }(-1)^{\alpha_d}
    +
    \sum_{a=1}^{d-1}
    \tr{ (|a\>\<a+1|) |j\>\<k|  }(-1)^{\alpha_a}
    \notag\\
    =& \delta_{j,1}\delta_{k,d} (-1)^{\alpha_d}
    +
    \sum_{a=1}^{d-1}
    \delta_{j,a+1}\delta_{k,a} (-1)^{\alpha_a}.
\end{align} 
Hence
\begin{align}
\sum_{j,k=1}^d w_{j,k} \tr{ U_{\bm \alpha} |j\>\<k|  }
&= w_{1,d} (-1)^{\alpha_d}
+ \sum_{a=1}^{d-1} w_{a+1,a}(-1)^{\alpha_a} .
\end{align}
 Now note that for any $j,k = 1,\dots,d$, we have 
\begin{align}
\frac{w_{j,k}}{2i} 
    &=
\frac{-i}{2} (\tr{X_k \rho X_j } - \tr{X_j \rho X_k  } )\notag\\
&=
\frac{1}{2} (\tr{X_k \rho (i X_j)^{\dagger} } + \tr{(iX_j) \rho X_k^\dagger  } ).
\end{align}
Hence we get
\begin{align}
  &\frac{1}{2i}  \sum_{j=1}^d \sum_{k=1}^d
    \tr{w_{j,k}
U_{\bm \alpha}  |j\>\<k|   }
\notag\\
=&
\frac{1}{2}\left(  \tr{ X_d \rho (iX_1)^\dagger } + \tr{ (iX_1) \rho X_d^\dagger}  \right) (-1)^{\alpha_d}
\notag\\
&+   
\sum_{k=1}^{d-1} 
\frac{1}{2}\left( \tr{ X_k \rho (iX_{k+1})^\dagger} + \tr{(iX_{k+1})^\dagger \rho X_k^\dagger  } \right)
    (-1)^{\alpha_k} .\label{lem:eq2}
\end{align}
Now 
\begin{align}
    &\frac{1}{2}\tr{
    (X_d + (-1)^{\alpha_d} i X_1)
    \rho
    (X_d + (-1)^{\alpha_d} i X_1)^\dagger
    }\notag\\
    =&
    \frac{1}{2} \tr{X_d \rho X_d^\dagger}
    +
    \frac{1}{2} \tr{X_1 \rho X_1^\dagger}
    \notag\\
    &+
    \frac{1}{2} \tr{X_d \rho
    ((-1)^{\alpha_d} iX_1)^\dagger}
    +
    \frac{1}{2} \tr{
    ((-1)^{\alpha_d} iX_1)
    \rho
    X_d^\dagger},\label{lem:eq3}
\end{align}
and similarly, for all $a=1,\dots , d-1$,
we have
\begin{align}
  &\frac{1}{2}\tr{
    (X_{a} + (-1)^{\alpha_a} i X_{a+1})
    \rho
    (X_a + (-1)^{\alpha_a} i X_{a+1})^\dagger
    }\notag\\
    =&
    \frac{1}{2} \tr{X_a \rho X_a^\dagger}
    +
    \frac{1}{2} \tr{X_{a+1} \rho X_{a+1}^\dagger}
    \notag\\
    &+
    \frac{1}{2} \tr{X_a \rho
    ((-1)^{\alpha_a} iX_{a+1})^\dagger}
    +
    \frac{1}{2} \tr{
    ((-1)^{\alpha_a} iX_{a+1})
    \rho
    X_a^\dagger}.\label{lem:eq4}
\end{align}
From \eqref{lem:eq1}, \eqref{lem:eq2}, \eqref{lem:eq3} and \eqref{lem:eq4}, the lemma follows.
\end{proof} 
To see how this Lemma works explicitly for three parameter ($d=3$) scenario, note that
\begin{align}
  & \tr{ {\rm Re}Z} + \frac i {2} \sum_{(a,b) \in E_{d,1} } w_{a,b}\notag\\
 =& \tr{X_1 \rho X_1^\dagger} +\tr{X_2 \rho X_2^\dagger} +\tr{X_3 \rho X_3^\dagger} \notag\\
 & + \frac i 2 (\tr{X_2 \rho X_1^\dagger} + \tr{X_3 \rho X_2^\dagger} + \tr{X_1 \rho X_3^\dagger} )
 \notag\\
 & -   \frac i 2 (\tr{X_1 \rho X_2^\dagger} + \tr{X_2 \rho X_3^\dagger} + \tr{X_3 \rho X_1^\dagger} )
 \notag\\
 =&
 \frac 1 2 \tr{(X_1+iX_2) \rho (X_1+iX_2)^\dagger} \notag\\
 &
 + \frac 1 2 \tr{(X_2+iX_3)\rho (X_2+iX_3)^\dagger}
 \notag\\
 &+ \frac 1 2 \tr{(X_3+iX_1)\rho (X_3+iX_1)^\dagger}.
\end{align}
Then we can rewrite \eqref{opt0-d} as an optimisation over Hermitian matrices $X_1, \dots, X_d$ where
\begin{mini} 
{X_1,\dots ,X_d}{
\max_{{\bm \alpha }\in\{0,1\}^d} V_{\bm \alpha}  ,}
{}{}
\addConstraint{\Tr{\rho X_j}}{=0}
\addConstraint{\Tr{\partial_j \rho X_k}}{=\delta_{jk}.}
\label{opt1-d}
\end{mini}
We can rewrite with an introduction of an auxiliary variable $t \in \mathbb R$ so that \eqref{opt1-d} is equivalent to 
\begin{mini} 
{X_1,\dots ,X_d,t}{
t ,}
{}{}
\addConstraint{\Tr{\rho X_j}}{=0}
\addConstraint{\Tr{\partial_j \rho X_k}}{=\delta_{jk}}
\addConstraint{V_{\bm \alpha} }{\le t}
\addConstraint{{\bm \alpha} }{\in\{0,1\}^d.}
\label{opt2-d}
\end{mini} 
This minimisation problem can be numerically checked for consistency with the optimisation in Eq.~\eqref{opt0-d}. 

 \subsection{Diagonalising the quadratic forms}
 
\noindent
Now note that the lower bound $V_{\bm \alpha}$ that we have for the objective function is a quadratic function of the optimization variables $X_1+i(-1)^{\alpha_1}X_2, \dots, X_{d-1}+i(-1)^{\alpha_{d-1}}X_d, X_d+i(-1)^{\alpha_d}X_1$, and these optimisation variables depend on the binary vector ${\bm \alpha}$.
 We can alternatively write $V_{\bm \alpha}$ in terms of optimization variables $Y_1,\dots,Y_d$ that are independent of ${\bm \alpha}$. We can quantify the linear dependence of the variables $Y_1,\dots, Y_d$ on the variables $X_1, \dots, X_d$ using the following matrix equation
\begin{align}
{\bf Y}    = {\bf S} {\bf X},
\end{align}
where
\begin{align}
{\bf S} = 
       \begin{pmatrix}
    S_{1,1} {\bf 1}& \dots & S_{1,d}  {\bf 1}\\
    \vdots &  & \vdots \\
    S_{d,1} {\bf 1}& \dots & S_{d,d} {\bf 1} \\
    \end{pmatrix} ,
    \quad 
    {\bf Y} =
    \begin{pmatrix}
    Y_1\\ \vdots \\ Y_d 
    \end{pmatrix},
    \quad
    {\bf X} = 
        \begin{pmatrix}
    X_1\\ \vdots \\ X_d 
    \end{pmatrix}
,
\end{align}
and 
\begin{align}
 S= \begin{cases}
    \sum_{j \in \mathbb Z_d} 
    \left( |j\>\<j|  + i |j\>\<j \oplus 1|   \right)  & d \neq 0 \pmod{4} \\
     \sum_{j \in \mathbb Z_d} 
    \left( |j\>\<j|  + (-1)^{\delta_{j,d}}i|j\>\<j \oplus 1|   \right) & \mbox{otherwise }  \label{def:S-matrix}.
    \end{cases}
\end{align}
When $d$ is not a multiple of 4, $S_{j,j} = 1$ and $S_{j,j+1}=i$ for all $j=1,\dots, d-1$, and $S_{d,d}=1, S_{d,1} =i$, and all other matrix elements of $S$ are zero. 
For instance, when $d=3$,
we have
\begin{align}
    S = 
    \begin{pmatrix}
    1 &  i & 0 \\
    0 & 1 & i \\
    i & 0 & 1 \\
    \end{pmatrix}.
\end{align}
but when $d=4$ we have
\begin{align}
    S = 
    \begin{pmatrix}
    1 &  i & 0 & 0 \\
    0 & 1 & i & 0 \\
    0 & 0 & 1 & i \\
    -i & 0 & 0 & 1 \\
    \end{pmatrix}.
\end{align}
Note that the only difference when $d$ is a multiple of 4 is that we flipped the sign of the bottom-left matrix element.

Given such a set of $Y_j$ variables, it then follows that
\begin{align}
V_{\bm \alpha}
    &=
\frac 1 2 \sum_{j=1}^d
\left( \delta_{0, \bar{\alpha}_j}\tr{Y_j \rho Y_j^\dagger} +  \delta_{1, \bar{\alpha}_j} \tr{Y_j^\dagger \rho Y_j}  \right).
\label{eq:Vt-defi-Ys}
\end{align}
where $\bar{\bm{\alpha}}=\bm{\alpha}$ when $d$ is not a multiple of 4, and when $d$ is a multiple of 4 then $\bar{\bm{\alpha}}$ differs from $\bm{\alpha}$ by simply flipping the last bit.

In the following proposition, we determine when the matrix $S$ is full rank.
\begin{proposition}
Let $d$ be a positive integer, and let $S$ be a matrix as defined in \eqref{def:S-matrix}. 
Then $S$ has full rank.
\end{proposition}
\begin{proof}
First, we consider the case when $d$ is not a multiple of 4. Since $S$ is a circulant matrix, its eigenvectors are the Fourier modes 
$|\phi_k\> = \sum_{j=0}^{d-1} \omega^{jk} |j\>$
where $k=0,1,\dots,d-1$, and $\omega = \exp(2\pi i /d)$ is a root of unity. 
The only way to get $S |\phi_k\> = {\bf 0}$ is to have $\omega^{k} = i$ for some integer $k$,
but this is only possible if $d$ divides 4. So if $d$ does not divide 4, we cannot have $S |\phi_k\> = {\bf 0}$, which implies that $S$ does not have any zero eigenvalues. Hence $S$ is full rank.

Next, we consider the case when $d$ is a multiple of 4.  The first $d-1$ rows of $S$ form an upper triangular matrix and are therefore linearly independent.  We just need to show that the last row is linearly independent from the rest. Let us denote the $j$th row of $S$ by $s_j$, where $j$ goes from 1 to $d$. Consider an arbitrary sum of the first $d-1$ rows of the form
\begin{align}
v =   \sum_{j=1}^{d-1} c_j s_j .
\end{align}
We wish to know if there exists a choice of constants $c_j$ such that $v = s_d$ where our definition of $S$ (when $d$ is a multiple of 4) has
\begin{align}
s_d = (-i,0,\ldots,0, 1) .
\end{align}
By setting the $j^{\mathrm{th}}$ element of $v$ equal to the $j^{\mathrm{th}}$ element of $s_d$, we obtain an equation for each $j$
\begin{align}
    c_1 & = -i \\
    i c_j + c_{j+1} & = 0 \\
    i c_{d-1} & = 1 
\end{align}
with the middle line holding for all $1 \leq j < d-1 $.  It is simple to confirm that there does not exist a solution to this set of equations.  In particular, we have the recursive equation $c_{j+1}=(-i)c_j$ with initial condition $c_1=(-i)$ and this solves to $c_{j}=(-i)^j$.  This entails $c_{d-1}=(-i)^{d-1}=i(-1)^d=i$ when $d$ is a multiple of 4.  However, this contradicts $i c_{d-1}=1$ and so no solution exists.  Therefore, the last row $S$ is linearly independent from the rest and the matrix is full-rank.
%
%
%
\end{proof}
From the above proposition, we see that whenever $d$ is not a multiple of 4, the matrix $S$ is full rank, which implies that $T=S^{-1}$ exists. In this scenario, we can write
\begin{align}
  T=
        \begin{pmatrix}
    T_{1,1} & \dots & T_{1,d}  \\
    \vdots &  & \vdots \\
    T_{d,1} & \dots & T_{d,d}  \\
\end{pmatrix},
\end{align}
and it follows that for every $k=1,\dots, d$,
we can express the Hermitian observables $X_1,\dots, X_d$ as linear combinations of the matrices $Y_1,\dots , Y_d$.
\begin{align}
    X_k = \sum_{\ell=1}^d T_{k,\ell} Y_\ell.
    \label{eq:X-to-Y}
\end{align}
Recall that $\rho_j = \partial_j\rho$ for $j=1,\dots,d $. From \eqref{eq:X-to-Y} and the Hermiticity of $X_k$, we recast the equality constraints in \eqref{opt2-d} as
\begin{align}
 c_{0,k} ({\bf Y}) &=   \frac 1 2 \sum_{\ell=1}^d \left( 
        T_{k,\ell} \tr{\rho Y_\ell}  +  T_{k,\ell}^* \tr{\rho Y_\ell^\dagger} 
   \right)= 0,\\
 c_{j,k} ({\bf Y}) &= \frac 1 2 \sum_{\ell=1}^d 
   \left(
        T_{k,\ell} \tr{\rho_j Y_\ell}  +  T_{k,\ell}^* \tr{\rho_j Y_\ell^\dagger}
   \right)-  \delta_{j,k}= 0.
\end{align}
Since the variables $Y_\ell$ are non-Hermitian in general, we need to impose additional constraints, namely the fact that the corresponding $X_k$ are Hermitian.
The Hermiticity of $X_k$ implies from \eqref{eq:X-to-Y} that 
\begin{align}
  \sum_{\ell=1}^d \left(T_{k,\ell} Y_\ell - T_{k,\ell}^* Y_\ell^\dagger \right)
  &= 0. \label{eq:hermitian-constraint-1}
\end{align}
The left side of \eqref{eq:hermitian-constraint-1} is in general an antihermitian matrix, and to make it Hermitian, we multiply both sides by $i$ to get
\begin{align}
 H_k({\bf Y}) = i \sum_{\ell=1}^d \left(T_{k,\ell} Y_\ell - T_{k,\ell}^* Y_\ell^\dagger \right)
  &= 0. \label{eq:hermitian-constraint-2}
\end{align}
With all these constraints, we recast the optimisation problem \eqref{opt2-d} as the following optimisation problem.
\begin{mini} 
{Y_1,\dots ,Y_d,t}{
t ,}
{}{}
\addConstraint{  c_{0,k}({\bf Y}) }{=0}
\addConstraint{  c_{j,k}({\bf Y}) }{=0}
\addConstraint{\frac 1 2 \sum_{\ell=1}^d
\left( \delta_{0,\alpha_\ell}\tr{Y_\ell \rho Y_\ell^\dagger} +  \delta_{1,\alpha_\ell} \tr{Y_\ell^\dagger \rho Y_\ell}  \right) }{\le t} 
\addConstraint{ H_k({\bf Y})  }{=0}
\addConstraint{{\bm \alpha} \in\{0,1\}^d}{.}
\label{opt3-d-Y}
\end{mini}

\subsection{Analysis on the Lagrangian}
\noindent
Here we consider the constraints in \eqref{opt3-d-Y} over $j,k=1,\dots, d$ and
$\alpha_1,\dots, \alpha_d = 0,1$, which gives us a total of $d(d+1)$ regular equality constraints, $d$ matrix equality constraints, and $2^d$ regular inequality constraints.
The Lagrangian corresponding to \eqref{opt3-d-Y} can then be written as 
\begin{align}
{\rm L}_d =&
   t
+ \sum_{j=0}^d \sum_{k=1}^d z_{j,k} c_{j,k}({\bf Y})
+ \sum_{k=1}^d \tr{\xi_k H_k({\bf Y}) } \notag\\
&+
 \frac 1 2 \sum_{{\bm \alpha}\in \{0,1\}^d} v_{\bm \alpha} \sum_{\ell=1}^d
\left( \delta_{0,\bar{\alpha}_\ell}\tr{Y_\ell \rho Y_\ell^\dagger} +  \delta_{1,\bar{\alpha}_\ell} \tr{Y_\ell^\dagger \rho Y_\ell}  \right) 
\notag\\
&- \sum_{{\bm \alpha}\in \{0,1\}^d} v_{\bm \alpha} t.
\label{eq:multiparameter-lagrangian}
\end{align}
Here, the Lagrange multipliers $z_{j,k}$ are real numbers while the Lagrange multipliers $v_{\bm \alpha}$ are non-negative numbers.
The Lagrange multipliers $\xi_k$ are Hermitian matrices in $\mathbb H_N$.
Note that the multiparameter Lagrangian is a quadratic form in ${\bf Y}$, and as such, can be minimised using Lemma \ref{lem:matrix-minimize}.
Before for we do so, we consider the minimisation of the Lagrangian with respect to the primal variable $t$.

If the Lagrangian multiplier $v_{\bm \alpha}$ do not all sum to one, by picking $t$ to either approach positive or negative infinity, the Lagrangian L$_d$ becomes unbounded. Hence the optimal multipliers $v_{\bm \alpha}$ must sum to one. 
By picking a discrete set of values of $v_{\bm \alpha}$ where $v_{\bm \alpha}$ is equal to zero to all but one value of ${\bm \alpha}$, and maximising the Lagrange dual function for each of these cases, we can obtain our lower bound to the multi-parameters \textsc{hcrb}.

Hence without loss of generality, there is some value of the binary vector ${\bm \alpha}$ for which the effective Lagrangian that we need to consider is 
\begin{align}
{\rm L}_{d,{\bm \alpha}} =& 
 \sum_{j=0}^d \sum_{k=1}^d z_{j,k} c_{j,k}({\bf Y})
+ \sum_{k=1}^d \tr{\xi_k H_k({\bf Y}) } \notag\\
&+
 \frac 1 2 \sum_{\ell=1}^d
\left( \delta_{0,\bar{\alpha}_\ell}\tr{Y_\ell \rho Y_\ell^\dagger} +  \delta_{1,\bar{\alpha}_\ell} \tr{Y_\ell^\dagger \rho Y_\ell}  \right)  .\label{eq:multiparameter-lagrangian-bft}
\end{align}
Now define 
\begin{align}
\Gamma_\ell &=
     \sum_{k=1}^d 
    T_{k,\ell}
    \left( 
        \sum_{j=0}^d z_{j,k} \rho_j  
        +
        i \xi_k
    \right).
\end{align}
By rewriting the terms on the first line on the right side
of~\eqref{eq:multiparameter-lagrangian-bft}, the effective Lagrangian becomes
\begin{align}
{\rm L}_{d,{\bm \alpha}} =& 
- \sum_{j=1}^d z_{j,j} + 
\frac 1 2 \sum_{\ell =1}^d
\left( 
    \tr{ \Gamma_\ell Y_\ell } + \tr{ \Gamma_\ell^\dagger Y_\ell^\dagger }
    \right.
    \notag\\
    &\quad + \left. \delta_{0,\bar{\alpha}_\ell}\tr{Y_\ell \rho Y_\ell^\dagger} +  \delta_{1,\bar{\alpha}_\ell} \tr{Y_\ell^\dagger \rho Y_\ell} \right).
\end{align}
Then, given that $\rho$ is a Hermitian full rank matrix, we can use Lemma \ref{lem:matrix-minimize} to get the corresponding Lagrange dual to be 
\begin{align}
g_{\bm \alpha}
=&
    \min_{{\bf Y}} {\rm L}_{d,{\bm \alpha}} \notag\\
=&
- \sum_{j=1}^d z_{j,j} 
- \sum_{\ell=1}^d
\frac{
\delta_{0,\bar{\alpha}_\ell} \tr{\Gamma_\ell \rho^{-1} \Gamma_\ell^\dagger}
+
\delta_{1,\bar{\alpha}_\ell} \tr{\Gamma_\ell^\dagger \rho^{-1} \Gamma_\ell}
}{2}.
\end{align}
 Our lower bound to the \textsc{hcrb} is thus 
\begin{align}
&\max_{{\bm \alpha} \in\{0,1\}^d }\max \{    g_{\bm \alpha}: z_{j,k}\in \mathbb R, \xi_k \in \mathbb H_D \},
\end{align}  
where $j=\{0,\dots, d\}, k=\{1,\dots, d\}$. Any feasible value of $g_{\bm \alpha}$ yields a lower bound to the \textsc{hcrb}.

\section{Minimising the Lagrangian}
\label{app:app_lag_min_arb_u}
\noindent
In this appendix, we extend our formalism to account for arbitrary values of the Lagrange dual variable $u$. When $u+v = 1$, we minimise the Lagrangian which we recall has the form
\begin{widetext}
\begin{align}
    \mathsf{L}(Y,u,{\bf z}) = -{\bf b}^\top {\bf z} 
    + u \tr{Y \rho Y^\dagger}
    + (1-u) \tr{Y^\dagger \rho Y}
    + \tr{A Y}
    + \tr{A^\dagger Y^\dagger},
\label{eqn:orig_Lagrangain}    
\end{align}
where 
\begin{align}
A = (z_1 \rho + z_2 \rho_1 + z_3 \rho_2 -    
i z_4 \rho -i z_5 \rho_1 -i z_6 \rho_2)/2.
\label{eq:A-mat}
\end{align}
Recall the notation $\rho_j=\partial_j\rho, j \in \{1,2\}$. The Frechet derivative of the Lagrangian in the matrix direction $H$ is given by
\begin{align}
    \nabla_Y(\mathsf{L}, H) = \lim_{h \to 0} \frac{ \mathsf{L}(Y + h H , u,{\bf z}) - \mathsf{L}(Y,u,{\bf z})  }{h}.
\end{align}
\begin{lemma}
Let $\mathsf{L}$ be the Lagrangian as defined in \eqref{eqn:orig_Lagrangain} and $A$ be as given in \eqref{eq:A-mat}.
Suppose that 
\begin{align}
    u Y \rho + (1-u) \rho Y + A^\dagger & =0.
    \label{eq:sylvester-equation}
\end{align}
Then $    \nabla_Y(L, H) = 0$.
\end{lemma}
\begin{proof}
Notice that
\begin{align}
    \nabla_Y(\mathsf{L}, H) &= 
      u ( \tr{H \rho Y^\dagger}
    +  \tr{Y \rho H^\dagger} )
    + (1-u) ( \tr{H^\dagger \rho Y}
    +  \tr{Y^\dagger \rho H} )
    + \tr{A H}
    + \tr{A^\dagger H^\dagger}. 
\end{align}
Now we use the cyclic property of the trace to write
\begin{align}
    \nabla_Y(\mathsf{L}, H) 
    &= 
      u( \tr{\rho Y^\dagger H}
    +  \tr{Y \rho H^\dagger} )
    + (1-u) ( \tr{ \rho Y H^\dagger}
    +  \tr{Y^\dagger \rho H} )
    + \tr{A H}
    + \tr{A^\dagger H^\dagger}
    \notag\\
    &=
      \tr{ B^\dagger H}
    + \tr{ B H^\dagger},
\end{align}
where
$B = u Y \rho + (1-u) \rho Y + A^\dagger$ 
and 
$B^\dagger = u \rho Y^\dagger + (1-u) Y^\dagger \rho + A$.
Since $B= 0$ by the assumption of our lemma, we must have $B^\dagger = 0$, and it follows that
$\nabla_Y(\mathsf{L}, H)=0$.
\end{proof}

Notice that Eq.~\eqref{eq:sylvester-equation} is a Sylvester equation, and solving it is a standard procedure, where a variant of the Bartel-Stewart algorithm can apply. When $\rho$ is a full-rank matrix,
we can solve this equation analytically. 
Let $\rho$ have the spectral decomposition $\rho = \sum_j p_j \ket{\smash{e_j}}\bra{\smash{e_j}}$ where $\ket{\smash{e_j}}$ are normalised eigenvectors of $\rho$.
In this case, \eqref{eq:sylvester-equation} is equivalent to 
\begin{align}
    u \sum_{j,k} \<e_j|Y|e_k\>|e_j\>\<e_k| p_k
    +
    (1-u) \sum_{j,k} p_j|e_j\>\<e_k| \<e_j|Y|e_k\>
    +
     \sum_{j,k}  |e_j\>\<e_k| \<e_j|A^\dagger |e_k\> = 0.
\end{align}
Simplifying this we get
\begin{align}
     \sum_{j,k} \bra{\smash{e_j}}Y\ket{e_k}
     ( u p_k + (1-u) p_j)
     \ket{\smash{e_j}}\bra{e_k} 
    =
     -\sum_{j,k}  \ket{\smash{e_j}}\bra{e_k} \bra{\smash{e_j}}A^\dagger \ket{e_k}, 
\end{align}
from which it follows that 
\begin{align}
   \<e_j|Y|e_k\>  =  -  ( u p_k + (1-u) p_j)^{-1} \<e_j|A^\dagger |e_k\>.
\end{align}
The following lemma then follows. 
\begin{lemma}
\label{lem:Ymin-v1}
The $Y$ that minimises the Lagrangian is given by 
\begin{align}
 Y=  
    -\sum_{j,k} 
  (u p_k + (1-u) p_j)^{-1}
  \bra{\smash{e_j}} A^{\dagger}\ket{e_k}
  \ket{\smash{e_j}}\bra{e_k}.   
\end{align}
\end{lemma} 
Crucially, we can write the $Y$ that minimises the Lagrangian as a linear combination of the Lagrange multipliers ${\bf z}$.
\begin{lemma}
\label{lem:Ymin-v2}
The $Y$ that minimises the Lagrangian is given by 
\begin{align}
    Y = 
    z_1 \gamma_1 +
    z_2 \gamma_2 + 
    z_3 \gamma_3 
    +i z_4 \gamma_1 
    +i z_5 \gamma_2 
    +i z_6 \gamma_3.
\label{eqn:ansatz_form}    
\end{align}
where the complex matrices
\begin{align}
 \gamma_1 &=  - {\bf 1}_d/2,\\
 \gamma_2 &=  - \sum_{j,k} 
  (u p_k + (1-u) p_j)^{-1}
  \<e_j| \rho_1/2 |e_k\>  |e_j\>\<e_k|,\\
 \gamma_3 &=  - \sum_{j,k} 
  (u p_k + (1-u) p_j)^{-1}
  \<e_j| \rho_2/2 |e_k\>  |e_j\>\<e_k|.
\end{align}
\end{lemma}
\begin{proof}
Recall the definition of $A$ in Eq.~\eqref{eq:A-mat}. Then, through Lemma~\ref{lem:Ymin-v1}, we find
\begin{align}
\begin{split}
 Y = &-
    (z_1 + i z_4)\sum_{j,k} 
  (u p_k + (1-u) p_j)^{-1}
  \<e_j| \rho/2 |e_k\>  |e_j\>\<e_k|
  \\
   &- (z_2 + i z_5)\sum_{j,k} 
  (u p_k + (1-u) p_j)^{-1}
  \<e_j| \rho_1/2 |e_k\>  |e_j\>\<e_k|
  \\
&-    (z_3 + i z_6)\sum_{j,k} 
  (u p_k + (1-u) p_j)^{-1}
  \<e_j| \rho_2/2 |e_k\>  |e_j\>\<e_k|.
\end{split}
\label{eqn:something}
\end{align}
Now we can make the simplification
\begin{align}
    \sum_{j,k} 
  (u p_k + (1-u) p_j)^{-1}
  \<e_j| \rho/2 |e_k\>  |e_j\>\<e_k|
  =
  \sum_{j} 
  (u p_j + (1-u) p_j)^{-1}
  p_j/2  |e_j\>\<e_j|  ={\bf 1}_d/2.
\end{align}
Hence, the result follows.
\end{proof}

By substituting the optimal value of $Y$ back into the Lagrangian, we find that Lagrangian is a quadratic in ${\bf z}$. Namely, we have the following.
\begin{lemma}
For fixed $u$ such that $0<u<1$, and where ${\bf z} = (z_1,\dots, z_6) \in \mathbb R^6$, the Lagrange dual of our Lagrangian is
\begin{align}
g(u,{\bf z})
=
    - {\bf b}^\top {\bf z}
+ {\bf z}^\top {\mathcal Q} {\bf z},
\end{align}
where
\begin{align}
    {\mathcal Q} = u {\rm Re}({\mathcal Q_1}) + (1-u) {\rm Re}({\mathcal Q_2})  + {\rm Re}({\mathcal Q_3}),
\end{align}
and
\begin{align}
{\mathcal Q_1} = 
 \begin{pmatrix}
   G_1 &  - i G_1 \\
 i G_1 &  G_1 \\
\end{pmatrix},
\quad
{\mathcal Q_2} = 
 \begin{pmatrix}
   G_2 &  i G_2 \\
- i G_2 &  G_2 \\
\end{pmatrix},
\quad
{\mathcal Q_3}=
    \frac{1}{2}
    \begin{pmatrix}
    G_3 + G_3^* & i(G_3 - G_3^*) \\
    -i(G_3 - G_3^*) & G_3 + G_3^* \\
    \end{pmatrix},
\label{eqn:gamma_J_matrices}
\end{align}
and
\begin{align}
G_1=      \begin{pmatrix}
1/4  &  0  & 0  \\
0
 &
 \tr{\gamma_2 \rho \gamma_2^\dagger}
 &
 \tr{\gamma_2 \rho \gamma_3^\dagger}
 \\
0
 &
 \tr{\gamma_3 \rho \gamma_2^\dagger}
 &
 \tr{\gamma_3 \rho \gamma_3^\dagger}\\
 \end{pmatrix},
 \quad
 G_2
 = \begin{pmatrix}
1/4  &  0  & 0  \\
0
 &
 \tr{\gamma_2^\dagger \rho \gamma_2}
 &
 \tr{\gamma_2^\dagger \rho \gamma_3 }
 \\
0
 &
 \tr{\gamma_3^\dagger \rho \gamma_2 }
 &
 \tr{\gamma_3^\dagger \rho \gamma_3}\\
 \end{pmatrix},
 \quad
 G_3
 = \begin{pmatrix}
-1/2  &  0  & 0  \\
0
 &
 \tr{\rho_1 \gamma_2}
 &
 \tr{\rho_1 \gamma_3 }
 \\
0
 &
 \tr{\rho_2 \gamma_2 }
 &
 \tr{\rho_2 \gamma_3}\\
 \end{pmatrix}
\end{align}
\label{lem:Ju_matrix_G_mats}
\end{lemma}
\begin{proof}
The Lagrange dual is given by substituting the optimal solution for $Y$ in the Lagrangian minimization. Recall the definition of the Lagrangian in Eq.~\eqref{eqn:orig_Lagrangain}, then the first term to evaluate is 
\begin{align}
\tr{Y \rho Y^\dagger}
    &=
 \tr{ ( z_1 \gamma_1 +
    z_2 \gamma_2 + 
    z_3 \gamma_3 
    +i z_4 \gamma_1 
    + i z_5 \gamma_2 
    + i z_6 \gamma_3)\rho
     ( z_1 \gamma_1^\dagger +
    z_2 \gamma_2^\dagger + 
    z_3 \gamma_3^\dagger 
    - i z_4 \gamma_1 ^\dagger
    - i z_5 \gamma_2 ^\dagger
    - i z_6 \gamma_3^\dagger)},
\end{align}
where we used the optimal solution for $Y$ as given in Lemma \ref{lem:Ymin-v2}. Writing this in matrix form, we have
\begin{align}
\tr{Y \rho Y^\dagger}
&=
{\bf z}^\top
\begin{pmatrix}
   G_1 &  - i G_1 \\
i G_1 &  G_1 \\
\end{pmatrix}
{\bf z} = {\bf z}^\top {\mathcal Q_1} {\bf z},
\end{align}
where {\bf z} is the column vector of Lagrange multipliers and the block matrix 
\begin{align}
    G_1=\begin{pmatrix}
 \tr{\gamma_1 \rho \gamma_1^\dagger} 
 &
 \tr{\gamma_1 \rho \gamma_2^\dagger}
 &
 \tr{\gamma_1 \rho \gamma_3^\dagger}
 \\
 \tr{\gamma_2 \rho \gamma_1^\dagger} 
 &
 \tr{\gamma_2 \rho \gamma_2^\dagger}
 &
 \tr{\gamma_2 \rho \gamma_3^\dagger}
 \\
 \tr{\gamma_3 \rho \gamma_1^\dagger} 
 &
 \tr{\gamma_3 \rho \gamma_2^\dagger}
 &
 \tr{\gamma_3 \rho \gamma_3^\dagger}\\
 \end{pmatrix}
 =
 \begin{pmatrix}
1/4
 &
 0
 &
 0
 \\
 0
 &
 \tr{\gamma_2 \rho \gamma_2^\dagger}
 &
 \tr{\gamma_2 \rho \gamma_3^\dagger}
 \\
 0
 &
 \tr{\gamma_3 \rho \gamma_2^\dagger}
 &
 \tr{\gamma_3 \rho \gamma_3^\dagger}\\
 \end{pmatrix}.
\end{align}
Here, we have used the definitions for $\gamma_j$ in Lemma~\ref{lem:Ymin-v2} and the traceless property of the state derivatives $\rho_1$ and $\rho_2$:
\begin{align}
 \tr{\rho \gamma_s}
 &=\tr{\rho \gamma_s^\dagger} =  \sum_{j} p_j
  ((1-u) p_j + u p_j)^{-1}
  \<e_j| \rho_{s-1}/2 |e_j\> 
  = \sum_{j} \<e_j| \rho_{s-1} |e_j\>/ 2 = 0,
\end{align}
for $s \in \{2,3\}$. Similarly, we determine the second term in the Lagrangian 
\begin{align}
\tr{Y^\dagger \rho Y}
    &=
 \tr{
 ( z_1 \gamma_1^\dagger +
    z_2 \gamma_2^\dagger + 
    z_3 \gamma_3^\dagger 
    - i z_4 \gamma_1 ^\dagger
    - i z_5 \gamma_2 ^\dagger
    - i z_6 \gamma_3^\dagger)
 \rho
    ( z_1 \gamma_1 +
    z_2 \gamma_2 + 
    z_3 \gamma_3 
    + i z_4 \gamma_1 
    + i z_5 \gamma_2 
    + i z_6 \gamma_3)}.    
\end{align}
In matrix form, this can similarly be written as
\begin{align}
\tr{Y^\dagger \rho Y}
&=
{\bf z}^\top
\begin{pmatrix}
   G_2 &  i G_2 \\
- i G_2 &  G_2 \\
\end{pmatrix}
{\bf z} = {\bf z}^\top {\mathcal Q_2} {\bf z}
\end{align}
where 
\begin{align}
G_2 =
 \begin{pmatrix}
1/4  &  0  & 0  \\
0
 &
 \tr{\gamma_2^\dagger \rho \gamma_2}
 &
 \tr{\gamma_2^\dagger \rho \gamma_3 }
 \\
0
 &
 \tr{\gamma_3^\dagger \rho \gamma_2 }
 &
 \tr{\gamma_3^\dagger \rho \gamma_3}\\
 \end{pmatrix}
\end{align}
Finally, we substitute the optimal solution for $Y$ into the last two terms of the Lagrangian
\begin{align}
\begin{split}
\tr{A Y + A^\dagger Y^\dagger} = \frac{1}{2} \tr{&( z_1 \rho +
    z_2 \rho_1 + 
    z_3 \rho_2 
    - i z_4 \rho
    - i z_5 \rho_1
    - i z_6 \rho_2)( z_1 \gamma_1 +
    z_2 \gamma_2 + 
    z_3 \gamma_3 
    + i z_4 \gamma_1 
    + i z_5 \gamma_2 
    + i z_6 \gamma_3) \\
    &+
    ( z_1 \rho +
    z_2 \rho_1 + 
    z_3 \rho_2 
    + i z_4 \rho
    + i z_5 \rho_1
    + i z_6 \rho_2)( z_1 \gamma_1^\dagger +
    z_2 \gamma_2^\dagger + 
    z_3 \gamma_3^\dagger 
    - i z_4 \gamma_1^\dagger 
    - i z_5 \gamma_2^\dagger 
    - i z_6 \gamma_3^\dagger)},
\end{split}
\end{align}
where we used the Hermicity of the state derivatives. Hence
\begin{align}
\tr{A Y + A^\dagger Y^\dagger} = \frac{1}{2} {\bf z}^\top
\begin{pmatrix}
   G_3 &  i G_3 \\
- i G_3 &  G_3 \\
\end{pmatrix}
{\bf z} + \frac{1}{2} {\bf z}^\top
\begin{pmatrix}
   G_4 &  - i G_4 \\
i G_4 &  G_4 \\
\end{pmatrix}
{\bf z} = {\bf z}^\top {\mathcal Q_3} {\bf z},
\end{align}
where 
\begin{align}
G_3 =
 \begin{pmatrix}
-1/2  &  0  & 0  \\
0
 &
 \tr{\rho_1 \gamma_2}
 &
 \tr{\rho_1 \gamma_3 }
 \\
0
 &
 \tr{\rho_2 \gamma_2 }
 &
 \tr{\rho_2 \gamma_3}\\
 \end{pmatrix},
 \quad
 G_4
 = \begin{pmatrix}
-1/2  &  0  & 0  \\
0
 &
 \tr{\rho_1 \gamma_2^\dagger}
 &
 \tr{\rho_1 \gamma_3^\dagger }
 \\
0
 &
 \tr{\rho_2 \gamma_2^\dagger }
 &
 \tr{\rho_2 \gamma_3^\dagger}\\
 \end{pmatrix}.
\end{align}
Notice that $G_4 = G_3^*$. Since the Lagrange dual must be real, and the dual variables $u$ and ${\bf z}$ must be real, we can take the real part of the matrices ${\mathcal Q_1}$, ${\mathcal Q_2}$ and ${\mathcal Q_3}$ to complete the proof.
\end{proof}

\end{widetext}

Because our optimization problem is convex, we are promised that $\mathsf{L}(Y,u,{\bf z})$ will be a concave function in both $u$ and $z$. 
This implies that for fixed $u$, 
$\mathsf{L}(Y,u,{\bf z}) $ is concave in ${\bf z}$, which implies that ${\mathcal Q}$ is negative definite. From the stationary point of $g(u, {\bf z})$ with respect to ${\bf z}$, we find that the optimal solution to ${\bf z}$ is given by the solution to the linear equation 
\begin{align}
    2 {\mathcal Q} {\bf z}= {\bf b}.
\end{align}
If ${\mathcal Q}$ is full rank, then we reach the optimal values for the Lagrange multipliers
\begin{align}
    {\bf z} = \frac{1}{2} {\mathcal Q} ^{-1} {\bf b}.
    \label{eq:zvec-analytical-solution}
\end{align}
Substituting this into the solution for $Y$ in Eq.~\eqref{eqn:ansatz_form} provides an ansatz that can saturate the \textsc{hcrb}, which is upper and lower bounded by
\begin{align}
\begin{split}
    \mathscr{U}_u &= \max\left\{{\bf z}^{\top} {\mathcal Q_1} {\bf z} , {\bf z}^{\top} {\mathcal Q_2} {\bf z}\right\},\\ 
    \mathscr{L}_u &= - {\bf b}^\top {\bf z} + {\bf z}^\top {\mathcal Q} {\bf z} 
\end{split}
\end{align}
respectively. For fixed $u$, our ansatz for $Y$ gives a tight bound when these two bounds are equivalent. If such a solution exists, we can optimise over the dual variable $u$, to find the optimal value. This can be determined numerically for any application in $\mathcal{O}(\text{polylog}(1/\epsilon))$ time, where $\epsilon$ is the duality gap. Alternatively, we can find the optimal $u$ by looking solely at the lower bound to the \textsc{hcrb}. 
\begin{align}
\begin{split}
    \mathscr{L}_u &= - \frac{1}{2}{\bf b}^\top {\mathcal Q} ^{-1} {\bf b}
    + 
    \frac{1}{4} {\bf b}^\top ({\mathcal Q} ^{-1} )^\top
    {\mathcal Q} {\mathcal Q} ^{-1} {\bf b}\\
    &=
     -\frac{1}{4} {\bf b}^\top  {\mathcal Q} ^{-1} {\bf b},
\end{split}
\end{align}
where we used Eq.~\eqref{eq:zvec-analytical-solution} and the fact that ${\mathcal Q}$ must be a symmetric matrix. The function $l_u$ is continuous and differentiable with respect to $u$. Also, duality theory of convex optimization promises that $\mathscr{L}_u$ is concave in $u$. Hence, in the scenario where the optimal $u$ is not attained for the values $u=\{0,1\}$, we will have that $\mathscr{L}_u$ is optimised at its stationary point $d \mathscr{L}_u/d u = 0$. The optimality condition of our \textsc{hcrb} is hence reduced to finding the roots of the stationary points of $\mathscr{L}_u$. This concludes our proof of theorem~\ref{thrm:hcrb_bounds_arb_u} in the main body of the text for any two-parameter estimation problem.


\section{Complexity analysis for Lagrange minimisation}
\label{subsec:complexity_analysis}

\noindent
For fixed $u$, our method to bound the \textsc{hcrb} requires minimising the Lagrangian. We found that this amounts to solving the Sylvester equation $A^\dagger = -(u Y \rho + (1-u) \rho Y)$ for $Y$. For large dimensional systems, evaluating the \textsc{hcrb} with concrete analytical results becomes increasingly cumbersome. In this scenario, a numerical approach can be used to handle the state diagonalisation to evaluate the \textsc{hcrb}. Here, we bound the complexity of our formalism to evaluating the \textsc{hcrb} numerically.

For large dimensional systems, the Sylvester equation is more efficiently solved by first vectorising the equation to 
\begin{align}
\vect(A^\dagger) = - (\mathbbm{1}_D\otimes (1-u)\rho + u\rho^\top\otimes\mathbbm{1}_D)\vect(Y),
\end{align}
and then solved using any system of linear equations solver. The Bartels-Stewart algorithm is an efficient and robust numerical solver for the Sylvester matrix for large-$D$~\cite{Bartels1972_ACM, Golub1979_IEE}, which outperforms well-known primitive implementations of Gaussian elimination. 
The complexity of the Bartels-Stewart algorithm scales as ${\mathcal{O}(D^3)}$. 

To circumvent any time complexity involved, we must have access to the basis that diagonalises the state. In this case, we solved the Sylvester equation analytically. For large $D$, there are efficient numerical methods that can attain the spectral decomposition of the state in sub-cubic time. For example, with $\rho(\bm{\theta}) \in \mathbb H_D$, the matrix inversion operation can be practically achieved using the Coppersmith-Winograd algorithm, which scales as $\mathcal{O}(D^{2.376})$~\cite{Coppersmith1990_JSC}. 

As a point of comparison, semi-definite programming (\textsc{sdp}) provide an alternative method to optimisation tasks. \textsc{sdp} programs can be applied to general problems and admit polynomial-time solvers, which highlight the power of this approach. The \textsc{hcrb} was recast as an \textsc{sdp} program in Ref.~\cite{Albarelli2019_PRL}. For a consistent complexity comparison with our method, we consider the non-trivial case where the state is full-rank. Hence, by observation of Eq. (11) in Ref.~\cite{Albarelli2019_PRL}, the variable $\bm{X}$ that is optimised has order $D^2$ terms. Further, notice that for each iteration of the \textsc{sdp} algorithm, the first constraint requires knowledge of the spectral decomposition of a matrix parameterised in terms of $\bm{X}$. Therefore, it is easy to observe that this brute force \textsc{sdp} approach has a time complexity greater than $\mathcal{O}(D^{2\times2.376})$. This indicates at least a quadratic speedup, which amounts to a significant improvement with increasing $D$.

\bibliographystyle{ieeetr}

\end{document}